\documentclass[pdflatex,sn-chicago]{sn-jnl}

\usepackage{graphicx}%
\usepackage{multirow}%
\usepackage{amsmath,amssymb,amsfonts}%
\usepackage{amsthm}%
\usepackage[scr=rsfs,scaled=0.96]{mathalfa}

\usepackage[title]{appendix}%
\usepackage{xcolor}%
\usepackage{textcomp}%
\usepackage{manyfoot}%
\usepackage{booktabs}%
\usepackage{algorithm}%
\usepackage{algorithmicx}%
\usepackage{algpseudocode}%
\usepackage{listings}%
\usepackage{lscape}

\theoremstyle{thmstyleone}%
\theoremstyle{thmstyletwo}%

\theoremstyle{thmstylethree}%

\theoremstyle{definition}
\newtheorem{Ass}{Assumption}
\newtheorem{thm}{Theorem}

\newtheorem{Remark}{Remark}%

\newtheorem{assump}{Assumption}

\begin{document}

\title[Short Title]{Focused Weighted-Average Least Squares Estimator}

\author*[1]{\fnm{Shou-Yung} \sur{Yin}}\email{syyin@mail.ntpu.edu.tw}

\affil[1]{\orgdiv{Department of Economics}, \orgname{National Taipei University}, \city{New Taipei}, \country{Taiwan}}

\abstract{
We propose a focused weighted-average least squares (FWALS) estimator that addresses the computational burden of focused model averaging. By semi-orthogonalizing auxiliary regressors, the weighting problem is reduced from $2^{k_2}$ sub-models to at most $k_2$ regressor-wise weights, yielding a tractable sub-optimal procedure. Under local-to-zero conditions, we derive the limiting distribution of FWALS for smooth focused functions and provide a plug-in AMSE criterion for data-driven weight selection. Simulations show that FWALS closely matches the focused information criterion (FIC) benchmark and delivers stable performance when focused function is designed for impulse response function. Prior-based WALS can be competitive in some settings, but its performance depends on the signal regime and the design of focused parameter. Overall, FWALS offers a practical and robust alternative with substantial computational savings.
}


\maketitle


\maketitle

\newpage
\section{Introduction}\label{Sec:Intro}

Model uncertainty has long been recognized as a central challenge in econometrics and applied statistics. In empirical research, investigators are often confronted with the difficulty of selecting among a wide range of candidate models, each of which may provide a plausible description of the data. Traditional model selection methods, such as those based on information criteria, are designed to select a single ``best'' model. However, it has been repeatedly documented that relying on a single selected model may ignore relevant information from other plausible specifications and lead to biased inference, particularly when the candidate models are of comparable quality. As a remedy, a large body of literature has developed around model averaging (MA) methods, which combine information across multiple models to deliver more robust parameter estimates and predictive performance; see, for example, \citet{steel2020model} for a comprehensive survey.

In the MA literature, one common approach is to define averaging estimators by minimizing a measurable asymptotic risk function for the entire parameter vector of the unconstrained model, with the risk evaluated with respect to estimators from candidate models. Representative studies include \citet{hansen2007least}, \citet{hansen2012jackknife}, \citet{zhang2023model-fa9}, \citet{zhang2023optimal-142}, and \citet{chen2025bregman-529}, among others. While these methods have been shown to be effective for whole model fit, empirical researchers are often primarily interested in a particular subset of structural parameters or functions of these parameters, rather than the nuisance parameters associated with control variables. This observation naturally motivates the development of focused model averaging methods, where the averaging weights are chosen to minimize the asymptotic risk of estimating a focused parameter vector, rather than the full parameter space.

The focused information criterion (FIC), originally proposed by \citet{claeskens2003focused} in the likelihood framework, provides a seminal foundation for this line of research. Subsequent work has extended the FIC idea to various frameworks, leading to a rich literature on focused model averaging; see, for example, \citet{hjort2006focused}, \citet{claeskens2008model}, \citet{zhang2012focused}, \citet{Liu2015}, \citet{Lu2015}, \citet{ditraglia2016using}, \citet{kitagawa2016model}, \citet{Lohmeyer2019}, and \citet{yin2021focused}. More recently, \citet{zhang2022unified} developed a unified representation of the asymptotic bias and variance, offering a general framework applicable to diverse settings. A recurring theme in this literature is the delicate task of accurately estimating the asymptotic bias and variance components, as these form the basis for both model selection and the determination of averaging weights. Despite important progress, implementing these methods in practice can be computationally demanding, particularly when the number of auxiliary regressors is large and the exponential growth of sub-models makes the calculation of weights numerically unstable.

The computational burden associated with traditional model averaging cannot be overlooked. For $k_2$ auxiliary regressors, the number of possible sub-models is $2^{k_2}$, and in order to evaluate the risk function or to calculate averaging weights, one must estimate all candidate sub-models. This requirement quickly becomes prohibitive even when $k_2$ is moderate, and the challenge is further compounded by the need to compute bias and variance terms for each sub-model in order to implement focused model averaging procedures. In practice, researchers may face a situation where evaluating the full model space is simply infeasible, and this limitation has motivated the search for computationally efficient alternatives.

Our paper builds on the orthogonalization idea of \citet{Magnus2010} and \citet{de2018weighted}, who proposed transforming auxiliary regressors to reduce correlations and thereby simplify the structure of model averaging estimators. By employing orthogonalized auxiliary regressors, we are able to recast the focused model averaging estimator into a form that depends only on a reduced number of weights associated with the auxiliary regressors, rather than weights defined over the full model space. This transformation drastically reduces the computational cost. Instead of evaluating exponentially many sub-models, we only need to work with at most $k_2$ terms, and the resulting weights can be interpreted as partial sums of the original model weights. While this construction does not yield the exact optimal solution implied by the full sub-model averaging, it provides a computationally attractive sub-optimal solution that remains closely aligned with the spirit of focused model averaging. As a result, the proposed approach facilitates empirical implementation in situation where traditional methods are infeasible due to computational constraints.

Our approach differs in important respects from two recent strands of the literature that also address this concern. First, \citet{charkhi2016minimum-d33} considered a minimum mean squared error model averaging estimator in likelihood models, where the averaging weights are derived from singleton equations and are only constrained to sum to one, without the non-negativity restrictions typically imposed in the MA literature. In contrast, our procedure requires that all weights lie within $[0,1]$, thereby ensuring interpretability as convex combinations and avoiding the instability that may arise from negative weights. Second, \citet{zhu2023scalable-a73} proposed a scalable frequentist model averaging method that employs a singular value decomposition to reduce the model space and proved the asymptotic equivalence of the minimum loss between the scalable and the traditional averaging estimator. Their focus was primarily on achieving statistical and computational efficiency without relying on local-to-zero assumption. Our contribution differs by targeting the focused parameter framework, where the primary interest lies in structural parameters and their transformations, and by demonstrating that orthogonalization-based weighting schemes can serve as a practical and theoretically coherent sub-optimal alternative to traditional FIC-based averaging.

In addition, our approach departs from the prior-based approaches in
\citet{Magnus2010,luca2022sampling-9be,luca2025bayesian-64a}. Prior-based approach employs posterior mean shrinkage induced by a chosen prior within the Normal location framework, which implies shrinkage weights for the coefficients. By contrast, our approach determines the weights by minimizing a plug-in asymptotic mean squared error~(AMSE) constructed for the focused parameter, so the weights are chosen to optimize the bias and variance trade-off of the focused estimand rather than a generic risk criterion.

Our simulation studies provide further support for the proposed method. In the baseline designs, the proposed approach, FIC from \citet{Liu2015}, and minimum mean squared error model averaging estimator from \citet{charkhi2016minimum-d33} perform comparably well, demonstrating that the transformation does not compromise the risk. However, when examining impulse response function (IRF) horizons, notable differences emerge. While our method and FIC remain closely aligned and exhibit stable risk performance across horizons, the approach suggested by \citet{charkhi2016minimum-d33} shows relatively unstable performance depending on the chosen horizon for IRF. These findings confirm that the proposed approach delivers stable performance and is computationally efficient, offering a strong alternative to traditional focused averaging estimators.

Taken together, the contribution of this paper is twofold. First, we extend the focused model averaging framework by introducing a weighted-average least squares estimator based on orthogonalized auxiliary regressors, which provides a computationally tractable solution for focused parameter estimation. Second, we clarify the relationship between our approach and existing methods in the literature, highlighting both the advantages of using the weights obtained from the AMSE and the role of focusing on structural parameters. 

The remainder of the paper is organized as follows. Section~\ref{Sec:Model} introduces the model specification and formalizes the averaging estimators. Section~3 introduces the proposed sub-optimal averaging estimator. Section~4 discusses its properties and investigates the limiting behavior. Section~5 presents simulation results to illustrate the computational and statistical advantages of the proposed method. Section~6 concludes.

\section{Model Specification}\label{Sec:Model}
In this section, we describe the model specification for the following discussion. We follow \citet{Magnus2010} and \citet{Liu2015}, and consider the regression framework as:
\begin{align}
y_i = \MBx_{i}^{\tp}\Bbeta+\epsilon_i= \MBx_{1i}^{\tp}\Bbeta_1+\MBx_{2i}^{\tp}\Bbeta_2+\epsilon_i,\quad i=1,...,N,\label{eq:main1}
\end{align}
where $y_i$ denotes the target variable of interest, $\MBx_{i}=[\MBx_{1i}^{\tp}\;\MBx_{2i}^{\tp}]^{\tp}$, and $\Bbeta=[\Bbeta_1^{\tp}\;\Bbeta_2^{\tp}]^{\tp}$. $\MBx_{1i}$ and $\MBx_{2i}$ represent the core regressors and the auxiliary regressors, respectively; $\epsilon_i$ is the random error and $N$ represents the sample size. The dimensions of $\MBx_{1i}$ and $\MBx_{2i}$ are $k_1$ and $k_2$ with the corresponding slope coefficients $\Bbeta_1$ and $\Bbeta_2$, and we further define that $k=k_1+k_2$ and $k$ is finite. In the literature, this model specification allows researchers to keep the core regressors for all possible sub-models by considering different combinations of the auxiliary regressors for statistical inference.

To have an easy interpretation for all sub-models estimation, we first rewrite Equation \eqref{eq:main1} as a matrix notation:
\begin{align}
\MBy = \MBX\Bbeta+\Bepsilon=\MBX_1\Bbeta_1+\MBX_2\Bbeta_2+\Bepsilon,
\end{align}
where $\MBy=[y_1\;\dots\;y_N]^{\tp}$, $\MBX=[\MBX_{1}\;\MBX_{2}]$, $\MBX_1=[\MBx_{11}\;\dots\;\MBx_{1N}]^{\tp}$, $\MBX_2=[\MBx_{21}\;\dots\;\MBx_{2N}]^{\tp}$ and $\Bepsilon=[\epsilon_1\;\dots\;\epsilon_N]^{\tp}$. Furthermore, we define a $k_{2m}\times k_2$ selection matrix $\BPi_m$ which selects the auxiliary regressors for sub-model $m$. For example, the included auxiliary regressors can be represented as $\MBX_{2m}=\MBX_2\BPi_m^{\tp}$. Let 
\begin{align}
\MBS_{m}=\begin{bmatrix}
\MBI_{k_1} & \MBzero_{k_1\times k_{2m}}\\
\MBzero_{k_2\times k_1} & \BPi_m^{\tp}
\end{bmatrix}.
\end{align}
Then we can define the least squares estimator of $\Bbeta=(\Bbeta_1^{\tp},\Bbeta_2^{\tp})^{\tp}$ for sub-model $m$:
\begin{align}
\hat{\Bbeta}_m=\MBS_m\left(\MBS_m^{\tp}\MBX^{\tp}\MBX\MBS_m\right)^{-1}\MBS_m^{\tp}\MBX^{\tp}\MBy.
\end{align}
The dimension of the above estimator is $k\times 1$.

Following \citet{Liu2015}, we define the parameter of interest by introducing a smooth real-valued function $\mu(\Bbeta_1)$. Accordingly, for any sub-model, we can obtain the estimate of this function $\hat{\mu}_m=\mu(\hat{\Bbeta}_{1m})$. Based on these estimates the averaging estimator for the focused parameter suggested by \citet{Liu2015} follows that
\begin{align}
\hat{\mu}(\MBw)=\sum_{m=1}^Mw_m\hat{\mu}_m,\label{eq:fic_avg}
\end{align}
where the weight vector, $\MBw=[w_1\;\dots\;w_M]^{\tp}$, satisfies the conditions:
\begin{align}
\mathcal{H}=\left\{\MBw\in[0,1]^{M}:\sum_{m=1}^Mw_m=1\right\},\label{eq:weights}
\end{align}
where $M$ denotes the number of total sub-models. To obtain the averaging estimator, we need to estimate $M$ sub-models. In the above case, $M = 2^{k_2}$. This number increases exponentially and results in substantial computational burden even when $k_2$ is moderate. In the next section, we introduce an alternative approach that substantially reduces computation time while preserving the advantages of the averaging procedure, as demonstrated in our simulation study.

\section{Sub-optimal Averaging Estimator}

\subsection{Focused Weighted-Average Least Squares Estimator}
In this section, we propose an averaging estimator of the focused parameter based on orthogonalized auxiliary regressors. The idea of orthogonalized auxiliary regressors is introduced by \citet{Magnus2010} and \citet{de2018weighted}, and this method defines new auxiliary regressors as
\begin{align}
\MBX_2^* = \MBX_2\hat{\BLambda}\hat{\MBP}^{-1/2},\label{eq:ortho}
\end{align}
where $\hat{\BLambda}=\DIAG\left(\DIAG\left(\frac{\MBX_2^{\tp}\MBM_1\MBX_2}{N}\right)\right)^{-1/2}$, $\hat{\MBP}=\hat{\BLambda}\frac{\MBX_2^{\tp}\MBM_1\MBX_2}{N}\hat{\BLambda}$ and $\MBM_1=\MBI_{N}-\MBX_1(\MBX_1^{\tp}\MBX_1)^{-1}\MBX_1^{\tp}$.
Given this semi-orthogonalization, following \citet{Magnus2010}, we can define the weighted-average least squares estimator~(WALS) of $\Bbeta_1$ as:
\begin{align}
\hat{\Bbeta}_{1\mathrm{WALS}}=\sum_{m=1}^Mw_m\hat{\Bbeta}_{1m},\label{eq:WALS}
\end{align}
and the sub-model estimate follows:
\begin{align}
\hat{\Bbeta}_{1m}=\hat{\Bbeta}_{1\mathrm{narrow}}-\hat{\BXi}\hat{\BLambda}\hat{\MBP}^{-1/2}(\BPi_m^{\tp}\BPi_m)\hat{\Bbeta}_2,\label{eq:b1m}
\end{align}
where $\hat{\BXi} = (\MBX_1^{\tp}\MBX_1)^{-1}\MBX_1^{\tp}\MBX_2$, $\hat{\Bbeta}_{1\mathrm{narrow}}= (\MBX_1^{\tp}\MBX_1)^{-1}\MBX_1^{\tp}\MBy$ and $\hat{\Bbeta}_2=\frac{\MBX_2^{*\tp}\MBM_1\MBy}{N}$.

The above results imply that the WALS for $\Bbeta_1$ defined in Equation \eqref{eq:WALS} can be rewritten as:
\begin{align}
\hat{\Bbeta}_{1\mathrm{WALS}} =& \hat{\Bbeta}_{1\mathrm{narrow}}-\hat{\BXi}\hat{\BLambda}\hat{\MBP}^{-1/2}\sum_{m=1}^Mw_m(\BPi_m^{\tp}\BPi_m)\hat{\Bbeta}_2\notag\\
                              =& \hat{\Bbeta}_{1\mathrm{narrow}}-\hat{\BXi}\hat{\BLambda}\hat{\MBP}^{-1/2}\tilde\MBW \hat{\Bbeta}_2\label{eq:WALS2}
\end{align}
where $\tilde\MBW=\sum_{m=1}^Mw_m(\BPi_m^{\tp}\BPi_m)$, and $\tilde\MBW$ is a diagonal matrix because of the property of orthogonalized auxiliary regressors. Accordingly, we can further obtain the following result that
\begin{align}
\tilde\MBW \hat{\Bbeta}_2 = \begin{bmatrix} \tilde{w}_1 \hat{\beta}_{21}&\cdots&\tilde{w}_{k_2} \hat{\beta}_{2k_2}\end{bmatrix}^{\tp},\label{WB2}
\end{align}
where $\tilde{w}_j$s for $j=1,...,k_2$ are the diagonal elements of $\tilde\MBW$. These diagonal elements play a different role of weights compared with $w_m$ defined in Equation \eqref{eq:weights} because $\tilde{w}_j$ is a partial sum of $w_m$s depending on the selection matrix $\BPi_m$. Based on the imposed conditions of $w_m$, we can also impose a weak condition on $\tilde{w}_j$s as:
\begin{align}
\tilde{\mathcal{H}}=\left\{\tilde\MBw\in[0,1]^{k_2}\right\}.\label{eq:weights2}
\end{align}
As discussed in \citet{Magnus2010}, while the WALS for $\Bbeta_1$ can be simplified as the estimator involving only $k_2$ $\tilde{w}_j$s, this approach does not consider the structure weights, $w_m$, and therefore it cannot be the optimal solution compared with the approach taking all possible sub-models into account. However, because of the number of weights to be optimized is $k_2$, it reduces the computation time especially when simulated confidence interval approach is used for providing the statistical inference.

Now we can define the focused WALS as follows:
\begin{align}
\hat{\mu}(\tilde\MBw) = \mu(\hat{\Bbeta}_{1\mathrm{WALS}}).\label{eq:fic_wals}
\end{align}
The key distinction between Equations \eqref{eq:fic_avg} and \eqref{eq:fic_wals} lies in the sequence by which the averaging estimator is constructed. In the focused WALS, we first form the averaging estimator of $\Bbeta_1$, denoted $\hat{\Bbeta}_{1\mathrm{WALS}}$, and then obtain the estimate of $\mu$ by directly applying the focused function to $\hat{\Bbeta}_{1\mathrm{WALS}}$. In contrast, the focus averaging estimator of \citet{Liu2015} computes the focused parameter within each sub-model and then averages these results across models. Although the procedures differ, their asymptotic behavior remains equivalent if the sub-models and the weights are the same. With the delta method, applying the focused function either before or after averaging does not alter the limiting distribution. Intuitively, this is because the focused function is smooth, so its linear approximation around the true parameter ensures that the order of applying averaging and transformation becomes irrelevant in large samples. Accordingly, the AMSE from the focused WALS can be treated as an alternative way to evaluate the model and obtain the weights. However, the weights formed in the proposed approach are different from \citet{Liu2015} so the asymptotic properties are also different. The formal results are established in the next section.

\begin{Remark}
The specific construction of the averaging estimator in \citet{Liu2015} prevents a similar reduction in computational burden, even if the auxiliary regressors are transformed via the semi-orthogonalization in Equation \eqref{eq:ortho}. Recall from Equation \eqref{eq:fic_avg} that for a given sub-model $m$, the estimator of the focused parameter is $\hat{\mu}_m=\mu(\hat{\Bbeta}_{1m})$, where $\hat{\Bbeta}_{1m}$ is defined in Equation \eqref{eq:b1m}. For any two distinct sub-models $m$ and $m'$, the structural forms of $\hat{\Bbeta}_{1m}$ and $\hat{\Bbeta}_{1m'}$ share the common base component $\hat{\Bbeta}_{1\mathrm{narrow}}$. Consequently, the sub-model estimators $\hat{\mu}_m$ and $\hat{\mu}_{m'}$ are inherently correlated through this shared component, despite the orthogonality of the auxiliary regressors themselves. Because these sub-model estimators are correlated, the cross-product terms do not vanish. As a result, the weight optimization problem can be simplified, leaving the $2^{k_2}$ computational burden unresolved in the framework adopted in \citet{Liu2015}. However, it is worth noting that this computational bottleneck is strictly tied to the evaluation of a focused parameter. If the objective were instead the global model fit with the core regressors always included considered in \citet{zhang2019inference}, the orthogonal transformation can be effective, thereby eliminating the combinatorial burden entirely.\footnote{When considering the global fit, the fitted value for sub-model $m$ can be decomposed using the Frisch-Waugh-Lovell theorem as $\hat{\MBy}_m = \MBX_1(\MBX_1^{\tp}\MBX_1)^{-1}\MBX_1^{\tp} \MBy + \sum_{j \in m} \MBM_1 \MBx_{2j}^* \hat{\beta}_{2j}^*$, where $\MBx_{2j}^*$ is the $j$th column of $\MBX_2^*$, and $\hat{\beta}_{2j}^*$ is the $j$th element of $\hat{\Bbeta}_2=\frac{\MBX_2^{*\tp}\MBM_1\MBy}{N}$. Because the orthogonalization ensures $\MBx_{2j}^{*\tp}\MBM_1\MBx_{2l}^* = 0$ for $j \neq l$, all cross-product terms in the risk criterion naturally vanish, reducing the optimization from $2^{k_2}$ combinations to $k_2$ combinations problem.}
\end{Remark}

\section{The Asymptotic Risk of Focused WALS}
In this section, we will discuss the limiting properties of the proposed focused WALS under the local-to-zero framework. Before going to the details of the asymptotic framework. We first state the assumptions made in this paper.
\begin{Ass}[Local-to-zero condition]\label{ass:loca}
$\Bbeta_2=\frac{\Bdelta}{\sqrt{N}}$.
\end{Ass}
\begin{Ass}[Sample moments convergence]\label{ass:lln}
$\frac{\MBX^{\tp}\MBX}{N}\CP\MBQ=\mathbb{E}(\MBx_i\MBx_i^{\tp})$, and $\MBQ$ can be partitioned as $\begin{bmatrix}\MBQ_{11}&\MBQ{12}\\ \MBQ_{21}&\MBQ_{22}\end{bmatrix}$ using the fact that $\MBX=[\MBX_{1}\;\MBX_{2}]$.
\end{Ass}
\begin{Ass}[Central limit theorem]\label{ass:clt}
$\frac{\MBX^{\tp}\Bepsilon}{\sqrt{N}}\CD\MBR\sim\mathrm{N}(\MBzero,\BOmega)$.
\end{Ass}

Assumption \ref{ass:loca} is a standard condition in the literature that facilitates the assessment of the trade-off between asymptotic bias and variance. We adopt Assumptions \ref{ass:lln} and \ref{ass:clt} as high-level assumptions for convenience. Furthermore, $\BOmega$ is assumed to be general to encompass cases of heteroskedasticity and mixing processes with different moments requirements in time series data.
Moreover, Assumption \ref{ass:lln} further implies that $\hat{\Blambda}$ and $\hat{\MBP}$ in Equation \eqref{eq:ortho} have well defined limits which follows that $\hat{\Blambda}\CP\DIAG\left(\DIAG\left(\MBQ_{22}-\MBQ_{21}\MBQ_{11}^{-1}\MBQ_{12}\right)\right)^{-1/2}=\BLambda$ and $\hat{\MBP}\CP\BLambda\left(\MBQ_{22}-\MBQ_{21}\MBQ_{11}^{-1}\MBQ_{12}\right)\BLambda=\MBP$.

\subsection{Asymptotic Mean Squared Error}
Under Assumptions \ref{ass:loca}-\ref{ass:clt}, and a non-stochastic $\tilde{\MBW}$, we have the following result as $N\rightarrow\infty$:
\begin{thm}\label{thm:beta}
Under Assumptions \ref{ass:loca}-\ref{ass:clt}, and a non-stochastic $\tilde{\MBW}$:
\begin{align*}
\sqrt{N}(\hat{\Bbeta}_{1\mathrm{WALS}}-\Bbeta_1)&\CD\BXi\MBC\left(\MBI-\tilde{\MBW}\right)\MBC^{-1}\Bdelta+\BPsi\MBR\\
&\equiv \MBR_{\Bbeta_1}(\tilde{\MBW})\\
&\sim\mathrm{N}\left(\BXi\MBC\left(\MBI-\tilde{\MBW}\right)\MBC^{-1}\Bdelta,\BPsi\BOmega\BPsi^{\tp}\right),
\end{align*}
and
\begin{align*}
\sqrt{N}\hat{\Bbeta}_2\CD\MBC^{-1}\Bdelta+\MBB\MBR,
\end{align*}
where $\BXi=p\lim_{N\rightarrow\infty}\hat{\BXi}=p\lim_{N\rightarrow\infty}\hat{\MBQ}_{11}^{-1}\hat{\MBQ}_{12}=\MBQ_{11}^{-1}\MBQ_{12}$, $\MBC=p\lim_{N\rightarrow\infty}\hat{\BLambda}\hat{\MBP}^{-1/2}=\BLambda\MBP^{-1/2}$, $\MBB=\left[-\MBC^{\tp}\BXi^{\tp}\quad\MBC^{\tp}\right]$ and $\BPsi = \left[\MBQ_{11}^{-1}+\BXi\MBC\tilde{\MBW}\MBC^{\tp}\BXi^{\tp}\quad-\BXi\MBC\tilde{\MBW}\MBC^{\tp}\right]$.
\end{thm}
The first result from the above theorem provides the asymptotic bias and variance of $\hat{\Bbeta}_{1\mathrm{WALS}}$ when the WALS is adopted under the local-to-zero framework. It is easy to observe that when $\tilde{\MBW}=\MBI_{k_2}$ the asymptotic bias is zero, and the WALS becomes the estimator considering all auxiliary regressors. If we let $\tilde{\MBW}=\MBzero$, the asymptotic bias becomes that $\BXi\MBC\MBC^{-1}\Bdelta=\MBQ_{11}^{-1}\MBQ_{12}\Bdelta$, and the asymptotic variance is $\MBQ_{11}^{-1}\BOmega_{11}\MBQ_{11}^{-1}$. These results imply the trade-off between the bias and variance via the choice of $\tilde{\MBW}$. Moreover, the asymptotic variance inherits the nice property from the WALS which does not need to involve the calculations between any two sub-models but only focus on the covariance between $\MBX_1$ and $\MBX_2^*$. The second result on the other hand shows that $\sqrt{N}\hat{\Bbeta}_2$ is an unbiased estimate of $\MBC^{-1}\Bdelta$ which plays an important role of recovering the information of the asymptotic bias of $\hat{\Bbeta}_{1\mathrm{WALS}}$.

Given the previous result, we can extend it by incorporating the focused function. Let $\MBD_{\Bbeta_1}=\frac{\partial\mu}{\partial\Bbeta_1}$ and assume the partial derivatives are continuous over all real values of $\Bbeta_1$. Then we can have the following theorem by applying the delta method.
\begin{thm}\label{thm:mu}
Under Assumptions \ref{ass:loca}-\ref{ass:clt}, suppose $N\rightarrow\infty$, we have
\begin{align*}
\sqrt{N}(\mu(\hat{\Bbeta}_{1\mathrm{WALS}})-\mu(\Bbeta_1))&\CD\MBD_{\Bbeta_1}^{\tp}\BXi\MBC\left(\MBI-\tilde{\MBW}\right)\MBC^{-1}\Bdelta+\MBD_{\Bbeta_1}^{\tp}\BPsi\MBR\\
&\equiv R_{\mu}(\tilde{\MBW}).
\end{align*}
\end{thm}
This theorem implies that the AMSE of the focused WALS $\mu(\hat{\Bbeta}_{1\mathrm{WALS}})$ is
\begin{align}
\mathrm{AMSE}(\mu(\hat{\Bbeta}_{1\mathrm{WALS}}))=&\MBD_{\Bbeta_1}^{\tp}\BXi\Bdelta\Bdelta^{\tp}\BXi^{\tp}\MBD_{\Bbeta_1}+\tilde{\MBw}^{\tp}\MBV\MBC^{-1}\Bdelta\Bdelta^{\tp}\MBC\MBV\tilde{\MBw}-2\tilde{\MBw}^{\tp}\MBV\MBC^{-1}\Bdelta\Bdelta^{\tp}\BXi^{\tp}\MBD_{\Bbeta_1}\notag\\
&+\MBD_{\Bbeta_1}^{\tp}\MBQ_{11}^{-1}\MBD_{\Bbeta_1}+\tilde{\MBw}^{\tp}\MBV\MBB\BOmega\MBB^{\tp}\MBV\tilde{\MBw}+2\tilde{\MBw}^{\tp}\MBV\MBB\BOmega\MBH\MBQ_{11}^{-1}\MBD_{\Bbeta_1},\label{eq:AMSE}
\end{align}
where $\MBV=\DIAG(\MBD_{\Bbeta_1}^{\tp}\BXi\MBC)=\DIAG(\MBC^{\tp}\BXi^{\tp}\MBD_{\Bbeta_1})$ and $\MBH=\left[\MBI\quad\MBzero\right]^{\tp}$; $\MBB$ has been defined in Theorem \ref{thm:beta}. The optimal $\tilde{\MBW}$, $\tilde{\MBW}^o$, can be obtained by minimizing the AMSE over $\tilde{\MBw}\in\tilde{\mathcal{H}}$ when $\BXi$, $\MBD_{\Bbeta_1}$, $\BLambda$, $\MBP$ and $\Bdelta$ are fixed, that is defined as
\begin{align}
\tilde{\MBw}^o=\arg\min_{\tilde{\MBw}\in\tilde{\mathcal{H}}}\mathrm{AMSE}(\mu(\hat{\Bbeta}_{1\mathrm{WALS}})),
\end{align}
and $\tilde{\MBW}^o=\DIAG(\tilde{\MBw}^o)$.

\subsection{Plug-in Focused WALS}
We have discussed the AMSE for the focused WALS; however, the optimal solution for $\tilde{\MBW}$ is infeasible because it depends on the unknown parameters including $\MBD_{\Bbeta_1}$, $\MBQ_{11}$, $\MBB$, $\MBC$, $\MBV$, $\BXi$, and $\Bdelta$ in Equation \eqref{eq:AMSE}. It is easy to observe that $\MBQ_{11}$, $\MBC$, $\BXi$ can be estimated consistently by $\hat{\MBQ}_{11}$, $\hat{\MBC}=\hat{\BLambda}\hat{\MBP}^{-1/2}$, and $\hat{\BXi}$, respectively. As for $\MBD_{\Bbeta_1}$ and $\MBV$, we can use the full model estimator $\hat{\Bbeta}_1=\hat{\Bbeta}_{1\mathrm{narrow}}-\hat{\BXi}\hat{\BLambda}\hat{\MBP}^{-1/2}\hat{\Bbeta}_2$ which is a special case of Equation \eqref{eq:WALS2} as $\tilde{\MBW}=\MBI$, and define $\MBD_{\hat{\Bbeta}_1}$ as $\MBD_{\Bbeta_1}$ evaluated at $\hat{\Bbeta}_1$. Because $\hat{\Bbeta}_1$ is consistent, so do $\MBD_{\hat{\Bbeta}_1}$ and $\hat{\MBV}=\DIAG(\MBD_{\hat{\Bbeta}_1}^{\tp}\hat{\BXi}\hat{\MBC})=\DIAG(\hat{\MBC}^{\tp}\hat{\BXi}^{\tp}\MBD_{\hat{\Bbeta}_1})$. 

However, we can only obtain an unbiased estimate of $\MBC^{-1}\Bdelta$ under the local-to-zero framework based on the second result in Theorem \ref{thm:beta}. Accordingly, we can further define $\hat{\Bdelta}=\sqrt{N}\hat{\MBC}\hat{\Bbeta}_2$, and it follows that:
\begin{align}
\hat{\Bdelta}&\CD\Bdelta+\MBC\MBB\MBR\equiv\MBR_{\Bdelta}\sim\mathrm{N}(\Bdelta,\MBC\MBB\BOmega\MBB^{\tp}\MBC^{\tp}),\label{eq:delta}\\
\hat{\Bdelta}\hat{\Bdelta}^{\tp}&\CD\Bdelta\Bdelta^{\tp}+\Bdelta\MBC\MBB\MBR+\MBR^{\tp}\MBB^{\tp}\MBC^{\tp}\Bdelta^{\tp}+\MBC\MBB\BOmega\MBB^{\tp}\MBC^{\tp}\equiv\BSigma_{\Bdelta}.\label{eq:deltadelta}
\end{align}
Equations \eqref{eq:delta} and \eqref{eq:deltadelta} provide the properties of estimated $\Bdelta$ which is supposed to be used to calculate the asymptotic bias. The first property implied by the above result is that $\hat{\Bdelta}$ is an unbiased estimate of $\Bdelta$, and the second property that we can obtain is that if we tend to calculate $\hat{\Bdelta}\hat{\Bdelta}^{\tp}$, it would not be an unbiased estimator of $\Bdelta\Bdelta^{\tp}$. Instead, it should be modified as:
\begin{equation}
\hat{\Bdelta}\hat{\Bdelta}^{\tp}-\hat{\MBC}\hat{\MBB}\hat{\BOmega}\hat{\MBB}^{\tp}\hat{\MBC}^{\tp}\CD\BSigma_{\Bdelta}-\MBC\MBB\BOmega\MBB^{\tp}\MBC^{\tp},\label{eq:ub_dd}
\end{equation}
and $\mathbb{E}(\BSigma_{\Bdelta}-\MBC\MBB\BOmega\MBB^{\tp}\MBC^{\tp})=\Bdelta\Bdelta^{\tp}$. Therefore, we can plug $\hat{\Bdelta}\hat{\Bdelta}^{\tp}-\hat{\MBC}\hat{\MBB}\hat{\BOmega}\hat{\MBB}^{\tp}\hat{\MBC}^{\tp}$ into Equation \eqref{eq:AMSE} directly and obtain the following estimator of AMSE for focused WALS:
\begin{align}
\widehat{\mathrm{AMSE}}(\mu(\hat{\Bbeta}_{1\mathrm{WALS}}))=&\MBD_{\hat{\Bbeta}_1}^{\tp}\hat{\BXi}\left(\hat{\Bdelta}\hat{\Bdelta}^{\tp}-\hat{\MBC}\hat{\MBB}\hat{\BOmega}\hat{\MBB}^{\tp}\hat{\MBC}^{\tp}\right)\hat{\BXi}^{\tp}\MBD_{\hat{\Bbeta}_1}\notag\\
&+\tilde{\MBw}^{\tp}\hat{\MBV}\hat{\MBC}^{-1}\left(\hat{\Bdelta}\hat{\Bdelta}^{\tp}-\hat{\MBC}\hat{\MBB}\hat{\BOmega}\hat{\MBB}^{\tp}\hat{\MBC}^{\tp}\right)\hat{\MBC}\hat{\MBV}\tilde{\MBw}\notag\\
&-2\tilde{\MBw}^{\tp}\hat{\MBV}\hat{\MBC}^{-1}\left(\hat{\Bdelta}\hat{\Bdelta}^{\tp}-\hat{\MBC}\hat{\MBB}\hat{\BOmega}\hat{\MBB}^{\tp}\hat{\MBC}^{\tp}\right)\hat{\BXi}^{\tp}\MBD_{\hat{\Bbeta}_1}+\MBD_{\hat{\Bbeta}_1}^{\tp}\hat{\MBQ}_{11}^{-1}\MBD_{\hat{\Bbeta}_1}\notag\\
&+\tilde{\MBw}^{\tp}\hat{\MBV}\hat{\MBB}\hat{\BOmega}\hat{\MBB}^{\tp}\hat{\MBV}\tilde{\MBw}+2\tilde{\MBw}^{\tp}\hat{\MBV}\hat{\MBB}\hat{\BOmega}\MBH\hat{\MBQ}_{11}^{-1}\MBD_{\hat{\Bbeta}_1}\notag\\
\CD &\MBD_{\Bbeta_1}^{\tp}\BXi\left(\BSigma_{\Bdelta}-\MBC\MBB\BOmega\MBB^{\tp}\MBC^{\tp}\right)\BXi^{\tp}\MBD_{\Bbeta_1}\notag\\
&+\tilde{\MBw}^{\tp}\MBV\MBC^{-1}\left(\BSigma_{\Bdelta}-\MBC\MBB\BOmega\MBB^{\tp}\MBC^{\tp}\right)\MBC\MBV\tilde{\MBw}\notag\\
&-2\tilde{\MBw}^{\tp}\MBV\MBC^{-1}\left(\BSigma_{\Bdelta}-\MBC\MBB\BOmega\MBB^{\tp}\MBC^{\tp}\right)\BXi^{\tp}\MBD_{\Bbeta_1}+\MBD_{\Bbeta_1}^{\tp}\MBQ_{11}^{-1}\MBD_{\Bbeta_1}\notag\\
&+\tilde{\MBw}^{\tp}\MBV\MBB\BOmega\MBB^{\tp}\MBV\tilde{\MBw}+2\tilde{\MBw}^{\tp}\MBV\MBB\BOmega\MBH\MBQ_{11}^{-1}\MBD_{\Bbeta_1}\notag\\
\equiv&\mathrm{AMSE}^*(\mu(\hat{\Bbeta}_{1\mathrm{WALS}})).
\label{eq:AMSE_sample}
\end{align}
Moreover, the above result also implies that $\mathbb{E}\left(\mathrm{AMSE}^*(\mu(\hat{\Bbeta}_{1\mathrm{WALS}}))\right)=\mathrm{AMSE}(\mu(\hat{\Bbeta}_{1\mathrm{WALS}})$, and therefore the estimated AMSE defined in Equation \eqref{eq:AMSE_sample} is an asymptotically unbiased estimator of the infeasible AMSE defined in Equation \eqref{eq:AMSE}.

Following \citet{Lu2015}, we can apply the argmax continuous mapping theorem to obtain the limiting distribution of the focused WALS given a data-driven $\hat{\MBw}$.
\begin{thm}\label{thm:plug}
Under Assumptions \ref{ass:loca}-\ref{ass:clt}, $\hat{\Bbeta}_{1\mathrm{WALS}}$ is calculated based on $\hat{\MBw}$ and $\hat{\MBw}=\arg\min_{\tilde{\MBw}\in\tilde{\mathcal{H}}}\widehat{\mathrm{AMSE}}(\mu(\hat{\Bbeta}_{1\mathrm{WALS}}))$, as $N\rightarrow\infty$, we have
\begin{align*}
&\hat\MBw\CD\tilde{\MBw}^*=\arg\min_{\tilde{\MBw}\in\tilde{\mathcal{H}}}\mathrm{AMSE}^*(\mu(\hat{\Bbeta}_{1\mathrm{WALS}})),\\
&\sqrt{N}(\hat{\Bbeta}_{1\mathrm{WALS}}-\Bbeta_1)\CD R_{\mu}(\tilde{\MBW}^*),\quad\tilde{\MBW}^*=\DIAG(\tilde\MBw^*).
\end{align*}
\end{thm}

\begin{Remark}
Our computational reduction relies on combining the semi-orthogonalization in \eqref{eq:ortho} with a weighting representation, which collapses the weight selection problem from the $2^{k_2}$ dimensional model space to $k_2$ scalar weights. This idea is related to the scalable averaging framework in \citet{zhu2023scalable-a73}, but the weight construction differs. \citet{zhu2023scalable-a73} consider Mallows model averaging \citep{hansen2007least} and Jackknife model averaging \citep{hansen2012jackknife} criteria, whereas our weights are obtained by minimizing a plug-in AMSE that is explicitly designed for the focused parameter under the local-to-zero condition in Assumption~\ref{ass:loca}. In the MMA/JMA literature, it is also common to restrict attention to a subset of candidate models, and the corresponding asymptotic optimality is established conditional on the restricted set. A similar restriction can be imposed in our framework, in which case the unbiased AMSE characterization and the resulting weights should be interpreted relative to the chosen candidate set. However, AMSEs computed under a restricted set are generally not directly comparable to those under the full set unless additional structure is imposed on $\MBX_2$ (e.g., approximate factor structure) or a screening step (e.g., sure independence screening) is used to justify the reduction in candidate models. We leave a formal analysis of such structured restrictions for future work.
\end{Remark}

\subsection{Comparison with Bayesian Approaches}
While the proposed approach shares the same spirit of orthogonalization adopted in \citet{Magnus2010}, we determine the sub-model weights in a fundamentally different way. To illustrate the main difference, let
\begin{align}
t_j = \frac{\hat{\beta}_{2j}}{\sigma_j},
\end{align}
where $\hat{\beta}_{2j}$ is the $j$th estimate from $\hat{\Bbeta}_2=\frac{\MBX_2^{*\tp}\MBM_1\MBy}{N}$. Suppose $\sigma_j$ is known. Under the distributional assumption used in \citet{Magnus2010} and \citet{luca2022sampling-9be}, or under the asymptotic framework discussed in \citet{luca2025bayesian-64a}, $t_j$ could be regarded as a Normal distribution with mean $\eta_j$ and unit variance. Given this result and a proper choice of prior density $\pi(\eta_j)$, we can apply Tweedie's formula \citep[e.g.,][]{pericchi1992exact} to obtain the posterior mean, $m(t_j)=\mathbb{E}(\eta|t_j)$, which is given by:
\begin{align}
m(t_j) = t_j + \frac{p'(t_j)}{p(t_j)},
\end{align}
where $p(t_j)=\int_{-\infty}^{\infty}\phi(t_j-\eta)\pi(\eta)d\eta$. Accordingly, we can define the WALS estimator with prior $\pi(\eta)$ as:
\begin{align}
\hat{\beta}_{2j,WALS-prior} =& m(t_j)\sigma_j\notag\\
=& \left( t_j + \frac{p'(t_j)}{p(t_j)} \right)\sigma_j\notag\\
=& \hat{\beta}_{2j}\left(1+\frac{p'(t_j)}{t_j p(t_j)}\right)\notag\\
=& \hat{\beta}_{2j}\omega_j.
\end{align}
This formulation matches the standard WALS estimator in \citet{Magnus2010}. By expressing it this way, we highlight that the Bayesian shrinkage acts exactly as a scalar weight $\omega_j$ applied to $\hat{\beta}_{2j}$. The term $\omega_j$ plays a role similar to the weight defined in Equation \eqref{WB2}. However, $\omega_j$ is obtained through the posterior mean, which serves to bound the minimax regret of $\hat{\beta}_{2j,WALS-prior}$ as suggested by \citet{Magnus2010}, \citet{luca2022sampling-9be}, and \citet{luca2025bayesian-64a}. In contrast, the weights in our approach are obtained by directly minimizing the AMSE, which calculates the exact sum of squared bias and variance. More importantly, the AMSE could be designed for the focused parameter in our approach.

To demonstrate the behavioral differences between the weights implied by the Bayesian posterior mean and our AMSE approach, we consider a simple data-generating process $y_i=x_i\beta+\epsilon_i$. We assume the local-to-zero framework $\beta=\frac{\delta}{\sqrt{N}}$, $\mathbb{E}(x_i^2)=1$, and $\mathbb{E}(\epsilon_i^2)=\sigma^2$. Consider the averaging estimator
\begin{align}
\hat{\beta}(w)=w\hat{\beta}_{OLS}+(1-w)\cdot 0,\qquad w\in[0,1],
\end{align}
where $\hat{\beta}_{OLS}$ is unbiased but has estimation variance, while the silly estimator equals zero and therefore has no estimation variance. Suppose we know the true value of $\delta$ in advance. Theoretically, the optimal weight based on AMSE takes the form $\frac{(\delta/\sigma)^2}{(\delta/\sigma)^2+1}$, while the weight implied by the Bayesian posterior mean is $\left(1+\frac{p'(\delta/\sigma)}{(\delta/\sigma) p(\delta/\sigma)}\right)$. 

In Figure \ref{fig:weights}, we plot the dynamics~(exclude the point when $\delta=0$) of the theoretical optimal AMSE weight alongside the weights implied by the four priors discussed in \citet{luca2025bayesian-64a}.\footnote{The four priors are Laplace, Weibull, Pareto, and Cauchy. The full definitions of these priors and their numerical integration are detailed in Appendix B.} As shown, as the signal-to-noise ratio ($\delta/\sigma$) approaches zero, the theoretically optimal AMSE weight strictly converges to zero. In contrast, as $\delta/\sigma \to 0$, the weights implied by the Bayesian posterior means remain positive, providing a milder, more conservative shrinkage. Furthermore, as $\delta/\sigma$ grows large, our AMSE approach rapidly converges to a weight of 1, relying fully on the oracle estimator. While the heavy-tailed priors (Cauchy, Weibull, Pareto) also eventually converge to 1, the standard Laplace prior retains a permanent shrinkage penalty, resulting in a constant bias that cannot be fully removed. Although this numerical example assumes knowledge of the true parameter, which is infeasible in practice, it clarifies the distinct mechanical differences between the two frameworks. We will further investigate how these differences translate to finite sample performance in the simulation section.

\begin{figure}[htbp]
    \centering
    \includegraphics[width=0.75\linewidth]{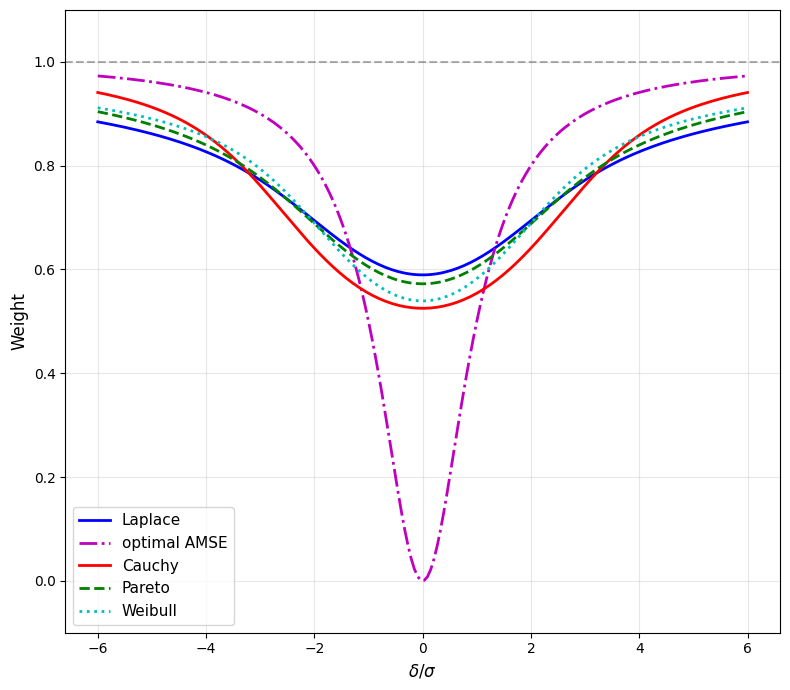}
    \caption{Weights: Theoretical AMSE vs. Bayesian Priors}
    \label{fig:weights}
\end{figure}

\begin{Remark}
\citet{luca2025bayesian-64a} derived the asymptotic distribution of WALS-type estimators under suitable conditions on the prior. In particular, under a local-to-zero specification, they showed that the resulting asymptotic distribution can exhibit substantial non-normality (e.g., skewness and excess kurtosis), and that the limit depends on the prior through the shrinkage rule implied by the posterior mean. Our approach departs from this Bayesian perspective. In our framework, the weights are chosen by minimizing a plug-in AMSE criterion that is explicitly constructed for the focused parameter (Theorem~\ref{thm:plug}), rather than being generated by a pre-specified prior. Consequently, the limiting distribution in Theorem~\ref{thm:plug} is driven by the AMSE optimal weighting mechanism, whereas the Bayesian-WALS limits vary with the prior through the induced shrinkage function. Therefore, even under the same local-to-zero assumption, the two procedures are governed by different weight-selection mappings (prior-induced shrinkage versus AMSE minimization). For practical implementation of WALS with prior, see \citet{luca2025weighteda,luca2025weightebd}, and the corresponding \textsf{R} implementation is available via the \texttt{WALS} package.
\end{Remark}

\section{Simulation Studies}
In this section, we conduct a Monte Carlo simulation to evaluate the performance of the proposed focused averaging method. Two key aspects are examined. We assess the risk by considering the mean squared error (MSE) of the focused parameter across different approaches against the MSE from the infeasible sub-model with two data-generating processes. We use the notation  $N$ to denote the sample size for basic design and $T$ for impulse response function design.

\subsection{Basic Design}
The data-generating process is generally based on the design in \citet{Liu2015} to facilitate comparison. We consider the following model specification: 
\begin{align} y_{i} = \MBx_{1i}^{\tp}\Bbeta_1+ \MBx_{2i}^{\tp}\Bbeta_2+\epsilon_i,\quad i=1,...,N, 
\end{align} 
where $\MBx_{i}=[\MBx_{1i}^{\tp}\;\MBx_{2i}^{\tp}]$ is drawn from a multivariate normal distribution with zero mean and covariance matrix $\BSigma_{\MBx}$, where the diagonal elements are one, and the off-diagonal elements are $\tau$. $\epsilon_i$ is drawn from a standard normal distribution. The dimensions of the core regressors, $\MBx_{1i}$, and the auxiliary regressors, $\MBx_{2i}$, are $k_1$ and $k_2$, respectively. The slope coefficients are specified as follows: \begin{align} 
\Bbeta_1 =& \frac{c_x}{a}[1\;\dots\;1]^{\tp},\\
\Bbeta_2 =& c_x\left[1\;\;\;\frac{k_2-1}{k_2}\;\dots\;\frac{1}{k_2}\right]^{\tp}, 
\end{align} 
where $a$ controls the relative importance of the core and auxiliary regressors. A larger $a$ implies smaller coefficients for the core regressors, making them relatively less important. The parameter $c_x$ governs the total explanatory power and controls the $R^2$ values through different choices of $c_x$.

In our first experimental design, we consider all combinations of $\tau=\{0.3,0.5,0.7\}$, $k_2=\{2,4,7\}$, $R^2=\{0.1,0.2,...,0.9\}$, $N=\{100,200\}$, given $a=12$, and $k_1=3$. This design is intended to examine how the small sample performance is affected by changes in the correlation between regressors and model complexity. The focused parameter considered in this example is defined as
\begin{align}
\mu =& \mu(\Bbeta_1)\notag\\
             =& \beta_{11}+\beta_{12}+\beta_{13}.
\end{align}

\subsection{Impulse Response Function Design}

This data-generating process is designed to examine the performance of the averaging estimators under different horizons of the impulse response function (IRF). We consider a sample size of $T=100$. The dependent variable $y_t$ is generated jointly with $k_2$ auxiliary regressors, where $k=k_1+k_2$ and $k_1=3$ is fixed as the number of autoregressive lags of $y_t$. Accordingly, we let $k\in\{5,7,10\}$ so that $k_2\in\{2,4,7\}$.

The dependent variable is generated from an autoregressive structure with $k_1=3$ lags:
\begin{align}
y_t &= \MBx_{1t}^{\tp}\Bbeta_1 + \MBx_{2t}^{\tp}\Bbeta_2 + u_t,\notag\\
&= \sum_{j=1}^{k_1}\beta_{1j} y_{t-j} + \MBx_{2t}^{\tp}\Bbeta_2 + u_t.\label{eq:dgp_y}
\end{align}
The autoregressive coefficients are given by
\begin{align}
\beta_{11}=0.5,\qquad \beta_{1j} = \frac{d}{\sqrt{T}(j-1)},\;\; j=2,3.\label{eq:dgp_phi}
\end{align}

The auxiliary regressors follow the stationary autoregressive process
\begin{align}
\MBx_{2t} = 0.2\,\MBx_{2,t-1} + \MBe_{\MBx,t},\label{eq:dgp_x}
\end{align}
where $\MBe_{\MBx,t}\sim\mathrm{N}(\MBzero,\BSigma_x$, with $\BSigma_x$ being a $k_2\times k_2$ covariance matrix whose diagonal elements are one and off-diagonal elements are $\tau$. The innovations $\MBe_{\MBx,t}$ are independent of the structural shocks $u_t$, where $u_t\sim \mathrm{N}(0,1)$.

For the auxiliary coefficients, we let  $s=\lfloor k_2/2\rfloor$ and adopt a near-sparse local specification:
\begin{align}
\Bbeta_2=\frac{c_y}{\sqrt{T}}\Btheta_y,\qquad 
\Btheta_y=
\begin{bmatrix}
\underbrace{1\;\;\cdots\;\;1}_{s\ \text{entries}}\;\;
\underbrace{0.05\;\;\cdots\;\;0.05}_{k_2-s\ \text{entries}}
\end{bmatrix}^{\tp},
\label{eq:dgp_beta2_sparse}
\end{align}
and the contribution of the auxiliary regressors is governed by the scaling constant $c_y$.

The initial values $(y_{-2},y_{-1},y_{0})$ are set to zero and $\MBx_{2,0}=\mathbf{0}$. After simulating $T+100$ observations, the first 100 are discarded and the last $T$ are retained as the effective sample. We consider all combinations of $c_y=\{0.1,..,4\}$ with 10 grids with $d=1$ and $\tau=0.2$. For each simulation, we evaluate the performance of the averaging estimators under different IRF horizons $h=\{1,3,5,7\}$.

The focused parameter considered in this design is the impulse response evaluated at different periods ($h$). More specifically, it is defined as
\begin{align}
\mu_{h} =& \mu_h(\Bbeta_1) \notag\\
=& \frac{\partial y_{t+h}}{\partial u_{t}} \notag\\
=& e_1^{\tp}A_{\beta}^{h}e_1, \label{eq:alpha-simulation2}
\end{align}
where $e_1$ is the first standard basis vector, defined as the vector with one in the first component and zeros elsewhere, $A_{\beta}^{h}=(A_{\beta})^h$, and $A_{\beta}$ denotes the companion matrix
\begin{align}
A_{\beta}=\begin{bmatrix}
\beta_{11} & \beta_{12} & \beta_{13} \\
1 & 0 & 0\\
0 & 1 & 0
\end{bmatrix}.\label{eq:companion_matrix}
\end{align}
The corresponding derivative of the focused parameter is given by
\begin{align}
\frac{\partial \mu_{h}}{\partial \beta_{j}} = \sum_{i=0}^h e_1^{\tp}A_{\beta}^{i}e_1 e_j^{\tp}A_{\beta}^{h-1-i}e_1,
\end{align}
where $e_j$ denotes the $j$th standard basis vector. Details are provided in \citet{Lohmeyer2019}.

\subsection{Alternative Methods}
In addition to the proposed focused WALS (denoted by $\mathsf{FWALS}$), we further provide results using the focused averaging estimator proposed by \citet{Liu2015} (denoted by $\mathsf{FIC}$), the smooth Akaike information criterion ($\mathsf{SAIC}$), the smooth Bayesian information criterion ($\mathsf{SBIC}$), the minimized MSE approach from \citet{charkhi2016minimum-d33} (denoted by $\mathsf{mMSE}$), and WALS with Laplace, Cauchy, Pareto, and Weibull priors ($\mathsf{WALS-Lap}$, $\mathsf{WALS-Cau}$, $\mathsf{WALS-Par}$, $\mathsf{WALS-Wei}$).\footnote{The AIC value for candidate model $m$ is calculated by the formula $AIC_m = N\ln{\left(\sum_{i=1}^N\hat{\epsilon}_{mi}^2/N\right)} + 2(k_1 + k_{2m})$, and $\hat{\epsilon}_{mi}$ and $k_{2m}$ denote the residuals and the number of selected auxiliary regressors, respectively. The weights for SAIC can be defined as $\hat{w}_m=\exp{\left(-AIC_m/2\right)}/\sum_{m=1}^M\exp{\left(-AIC_m/2\right)}$.}\footnote{The BIC value for candidate model $m$ is calculated by the formula $BIC_m = N\ln{\left(\sum_{i=1}^N\hat{\epsilon}_{mi}^2/N\right)} + \ln{(N)}(k_1 + k_{2m})$, and $\hat{\epsilon}_{mi}$ and $k_{2m}$ denote the residuals and the number of selected auxiliary regressors, respectively. The weights for SBIC can be defined as $\hat{w}_m=\exp{\left(-BIC_m/2\right)}/\sum_{m=1}^M\exp{\left(-BIC_m/2\right)}$.}\footnote{The details of the priors can be found in Appendix B.} We conduct our simulations with 1{,}000 replications for all cases.

\subsection{Small Sample Properties}
We summarize the simulation results for the first experiment in Figures \ref{fig1} to \ref{fig6}. Each figure displays results for three different values of the number of auxiliary regressors, and, for a given number of auxiliary regressors, we present the risk dynamics in relation to $R^2$ across all methods.

We first discuss the case when $N=100$. In Figure \ref{fig1}, regarding the non-WALS type approaches, we observe that, in most cases, $\mathsf{FWALS}$, $\mathsf{FIC}$ and $\mathsf{mMSE}$ exhibit very similar performance, regardless of the number of auxiliary regressors, when the correlation between regressors is relatively weak. This suggests that the transformation in the proposed method does not negatively affect the asymptotic MSE of the focused parameter. When compared with $\mathsf{SAIC}$ and $\mathsf{SBIC}$, both $\mathsf{FWALS}$, $\mathsf{FIC}$ and $\mathsf{mMSE}$ generally perform better, except when $R^2$ is small and $k_2$ is large. Additionally, $\mathsf{SBIC}$ tends to perform poorly (higher risk) when the true model is complex with many weak predictors, likely due to its preference for simpler models over focusing on the accuracy of the focused parameter. On the other hand, $\mathsf{SAIC}$ performs better, likely because it places more emphasis on relatively complex models. 

As for the WALS-type approaches, regardless of the choice of priors, we observe that $\mathsf{WALS-Lap}$, $\mathsf{WALS-Cau}$, $\mathsf{WALS-Par}$, and $\mathsf{WALS-Wei}$ deliver similar patterns. When $R^2$ is large, WALS with priors perform worse. This is mainly because the design is not used to minimize the risk of the focused parameter. In addition, $\mathsf{WALS-Lap}$ shows the potential drawback of permanent shrinkage (weight less than $1$) as $R^2$ is large, which could bias the focused parameter. In general, we observe that increasing $k_2$ further demonstrates the advantages of using focus-based methods.

Next, we examine the results for different values of $\tau$, as shown in Figures \ref{fig2} and \ref{fig3}. A similar pattern is evident: $\mathsf{FWALS}$, $\mathsf{FIC}$ and $\mathsf{mMSE}$ generally perform better, especially when $\tau=0.7$, $k_2=7$, and $R^2$ is higher. Additionally, we observe that as $\tau$ increases, $\mathsf{FWALS}$ slightly outperforms $\mathsf{FIC}$ when $R^2$ is relatively small. This can be explained by the advantage of the transformation in isolating the contributions of the transformed regressors. However, this advantage diminishes as $R^2$ increases, since higher $R^2$ implies that all auxiliary regressors have a greater impact on the focused parameter. In such cases, imposing more weight restrictions through $\mathsf{FWALS}$ becomes a disadvantage, and considering all combinations yields better results. This trade-off highlights a key aspect of the method. Overall, the simulation results indicate that the differences between $\mathsf{FWALS}$ and $\mathsf{FIC}$ are small across all cases discussed. Furthermore, when $\tau$ is large, the disadvantage of using the pre-specified priors is still obvious as $R^2$ increases. However, we observe that the performance of $\mathsf{FWALS}$, $\mathsf{FIC}$, and $\mathsf{mMSE}$ is slightly poor when $R^2$ is small and $\tau$ is large. This can be regarded as the price we need to pay because we need to estimate the bias, which can distort the estimate of the AMSE and therefore lead to unsatisfactory results. Similar patterns can be found when we consider the case $N=200$, as shown in Figures \ref{fig4} to \ref{fig6}.

In Figures \ref{fig7} to \ref{fig9}, we report the risk performance for different horizons $h$ of the impulse response function under several values of $k_2$. Across all horizons, $\mathsf{FWALS}$ and $\mathsf{FIC}$ exhibit very similar risk dynamics, which confirms that the proposed transformation does not sacrifice risk relative to the traditional focused averaging estimator while delivering substantial computational savings. When $k_2=2$, as shown in Figure \ref{fig7}, almost all methods have comparable performance across the range of horizons considered. The story changes as $k_2$ becomes larger. As $k_2$ increases and the number of relatively weak auxiliary regressors grows, we again find that $\mathsf{FWALS}$ and $\mathsf{FIC}$ are comparable to, or better than, the prior-based WALS approaches when $c_y$ is large (Figures \ref{fig8}--\ref{fig9}). When the signals from the auxiliary regressors become weak (small $c_y$), $\mathsf{FWALS}$ and $\mathsf{FIC}$ may suffer from an imprecise bias estimate, as observed in the first design; overall, however, the two methods still deliver similar performance. In contrast, $\mathsf{mMSE}$ shows unstable performance across horizons. This may reflect a drawback of allowing a broad weight space based on singleton equations, especially when the focused function is nonlinear. Finally, $\mathsf{SBIC}$ and $\mathsf{SAIC}$ display performance patterns broadly similar to $\mathsf{FWALS}$, $\mathsf{FIC}$, and the prior-based WALS methods, although $\mathsf{SBIC}$ appears slightly more volatile in terms of risk.

Taken together, these simulation results show that $\mathsf{FWALS}$ is a competitive alternative to $\mathsf{FIC}$ across horizons and values of $k_2$, combining computational efficiency with stable statistical performance, and providing more stable performance than $\mathsf{mMSE}$. Prior-based $\mathsf{WALS}$ methods can also be competitive in some settings, especially when the signals from the auxiliary regressors are small.

\begin{figure}[!htpb]
    \centering
    \includegraphics[width=0.32\linewidth]{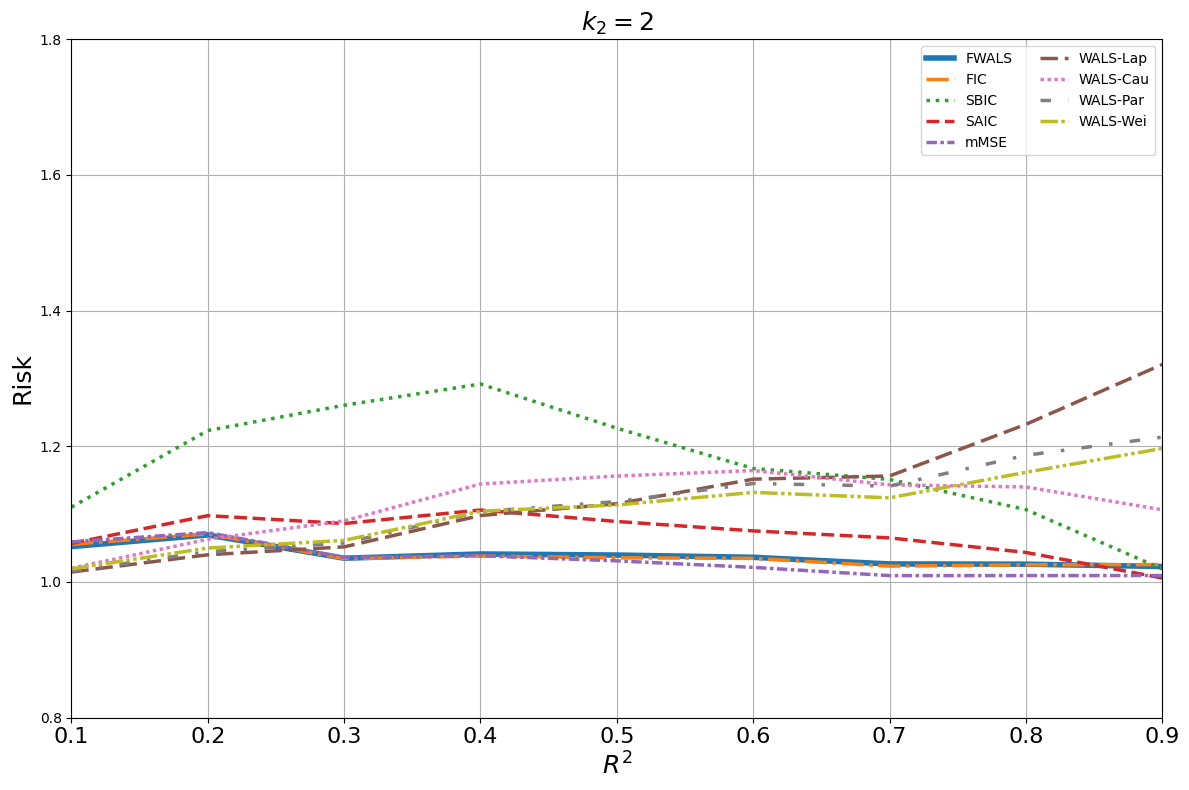}
    \includegraphics[width=0.32\linewidth]{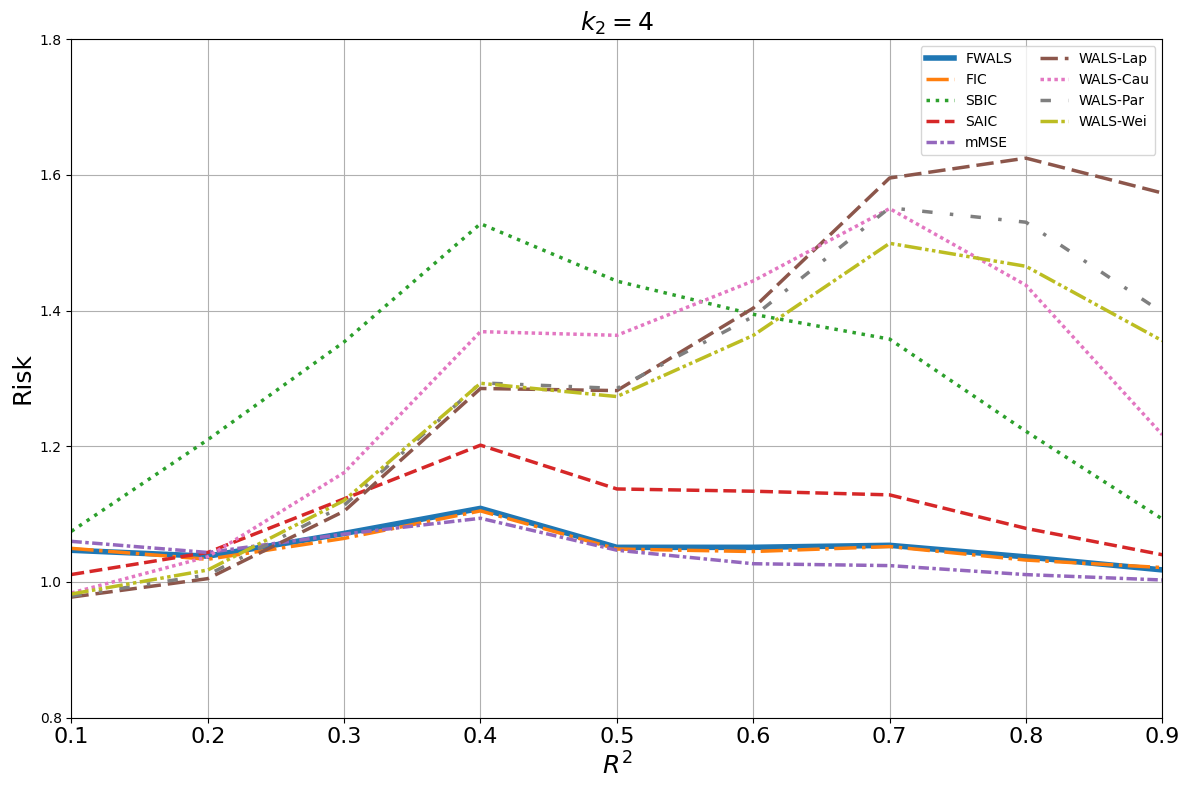}
    \includegraphics[width=0.32\linewidth]{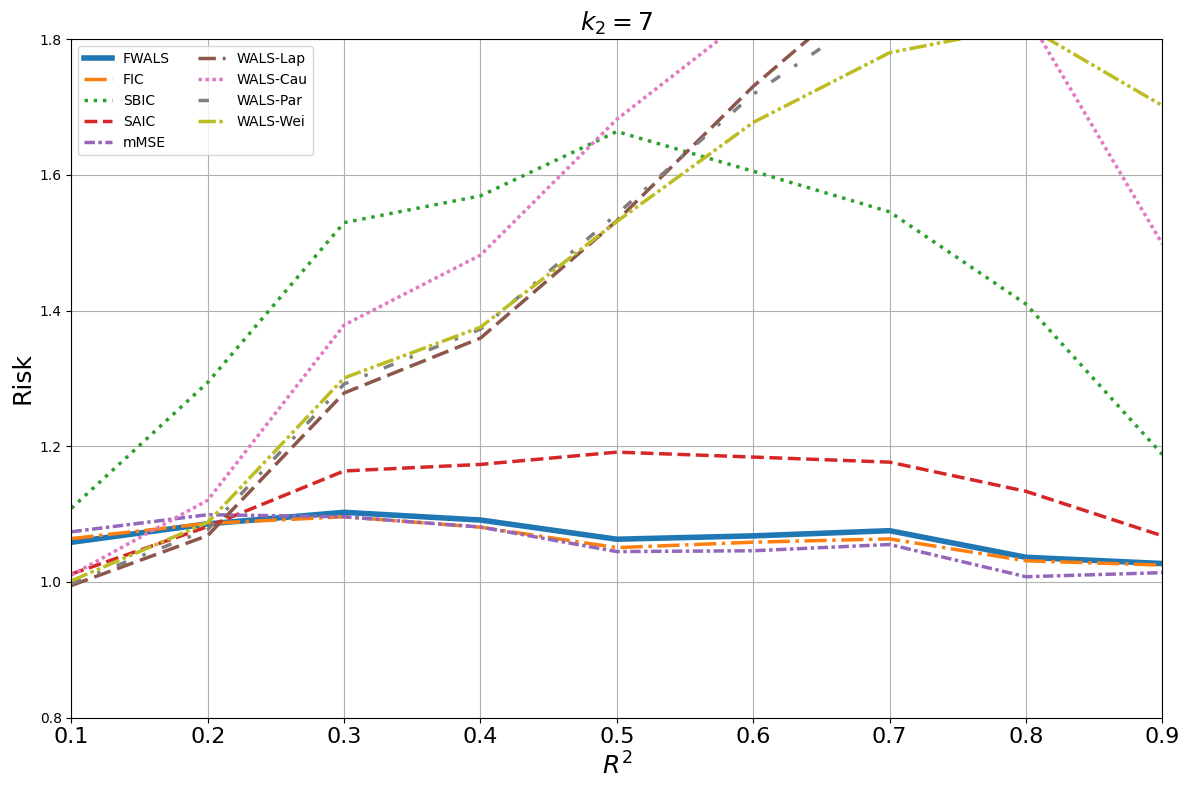}
    \caption{Risk with Different $k_2$s~($N=100$,$\tau=0.3$)}
    \label{fig1}
\end{figure}

\begin{figure}[!htpb]
    \centering
    \includegraphics[width=0.32\linewidth]{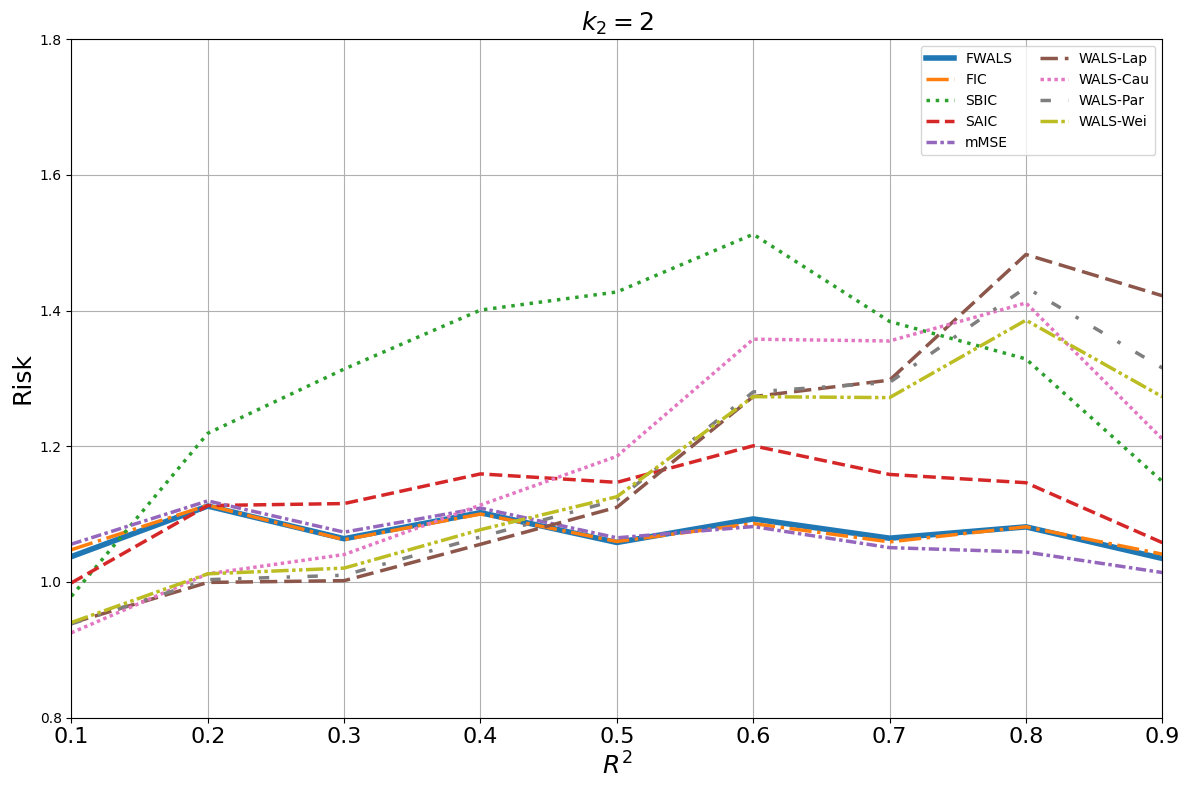}
    \includegraphics[width=0.32\linewidth]{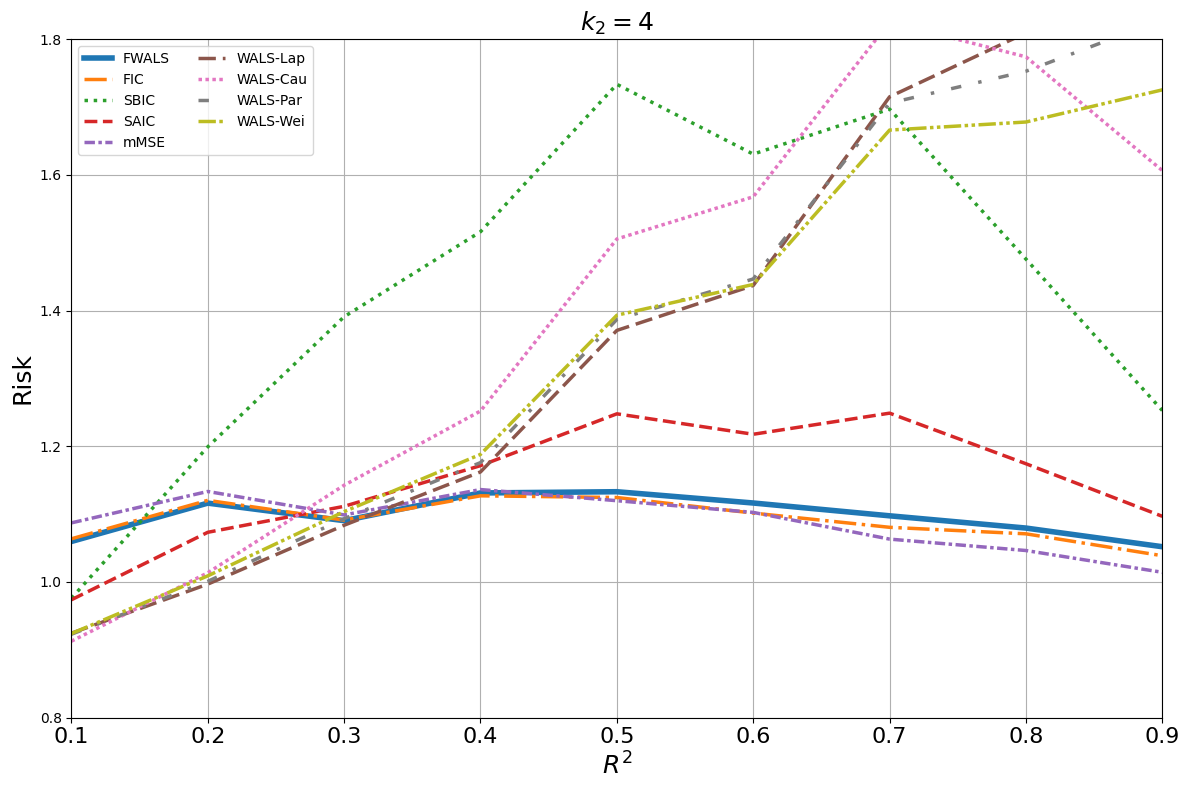}
    \includegraphics[width=0.32\linewidth]{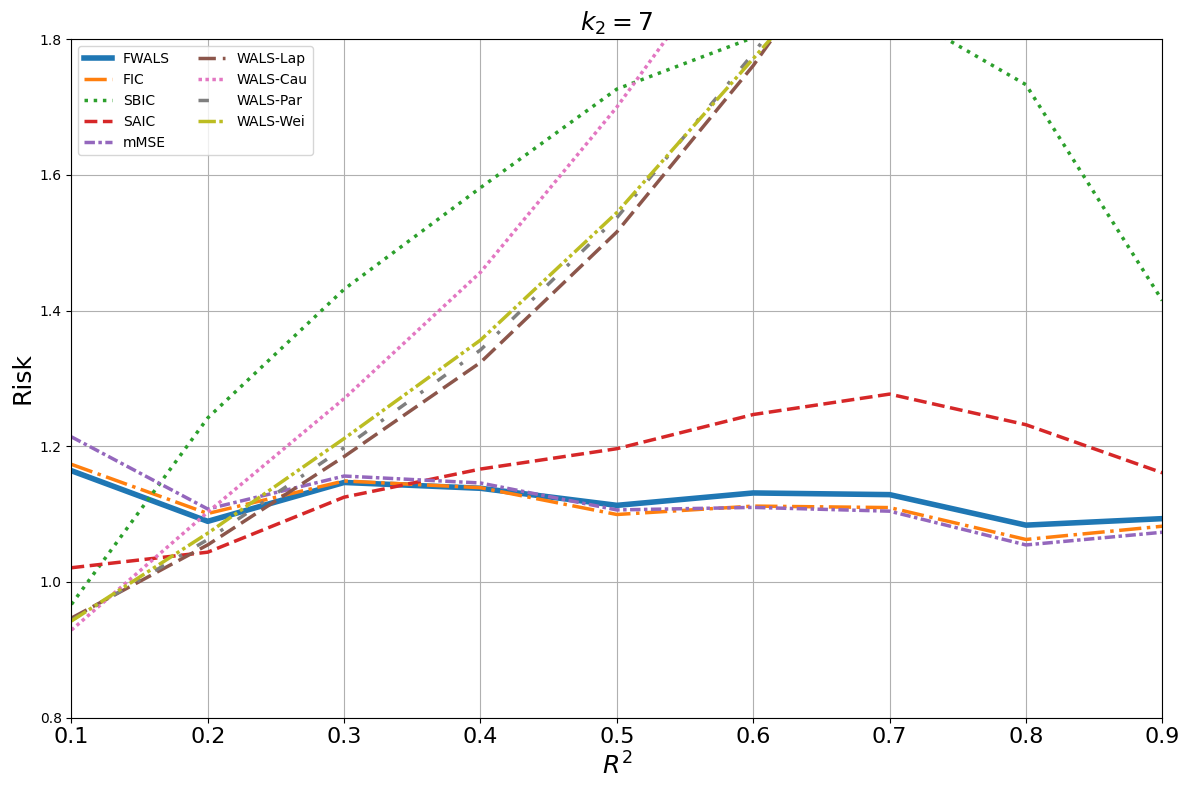}
    \caption{Risk with Different $k_2$s~($N=100$,$\tau=0.5$)}
    \label{fig2}
\end{figure}

\begin{figure}[!htpb]
    \centering
    \includegraphics[width=0.32\linewidth]{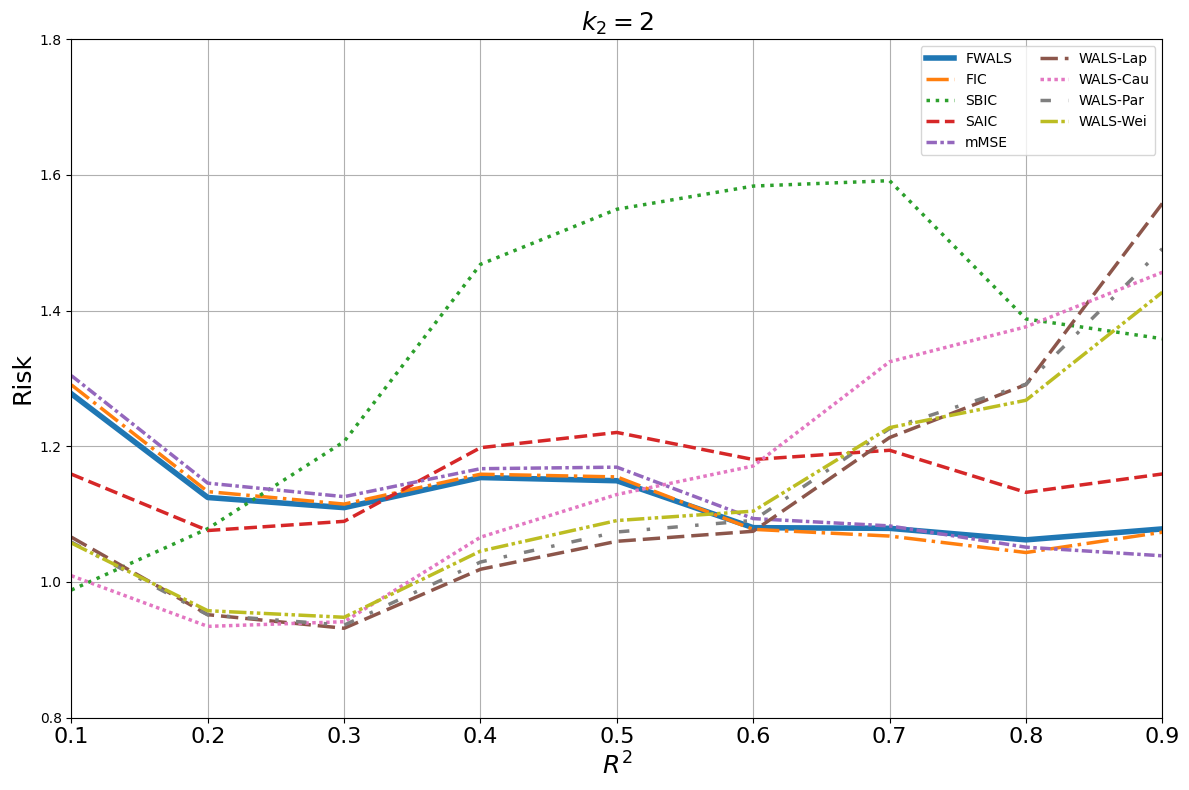}
    \includegraphics[width=0.32\linewidth]{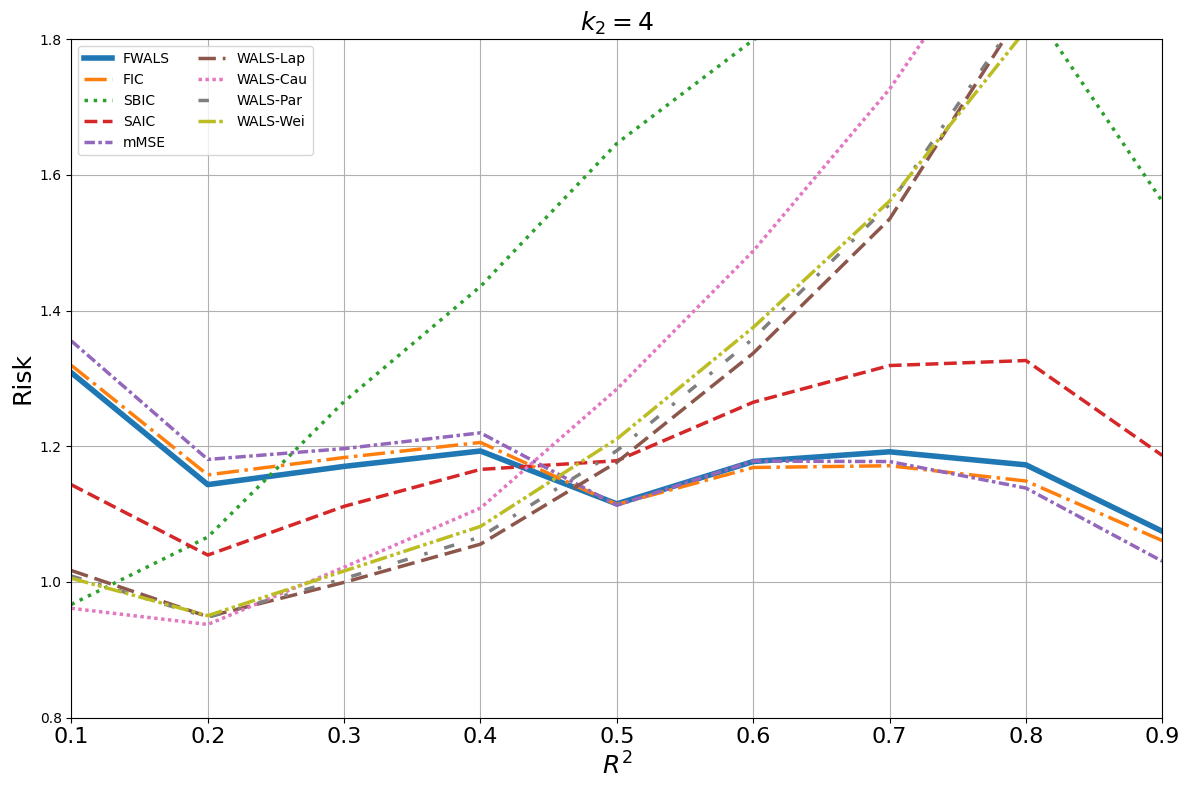}
    \includegraphics[width=0.32\linewidth]{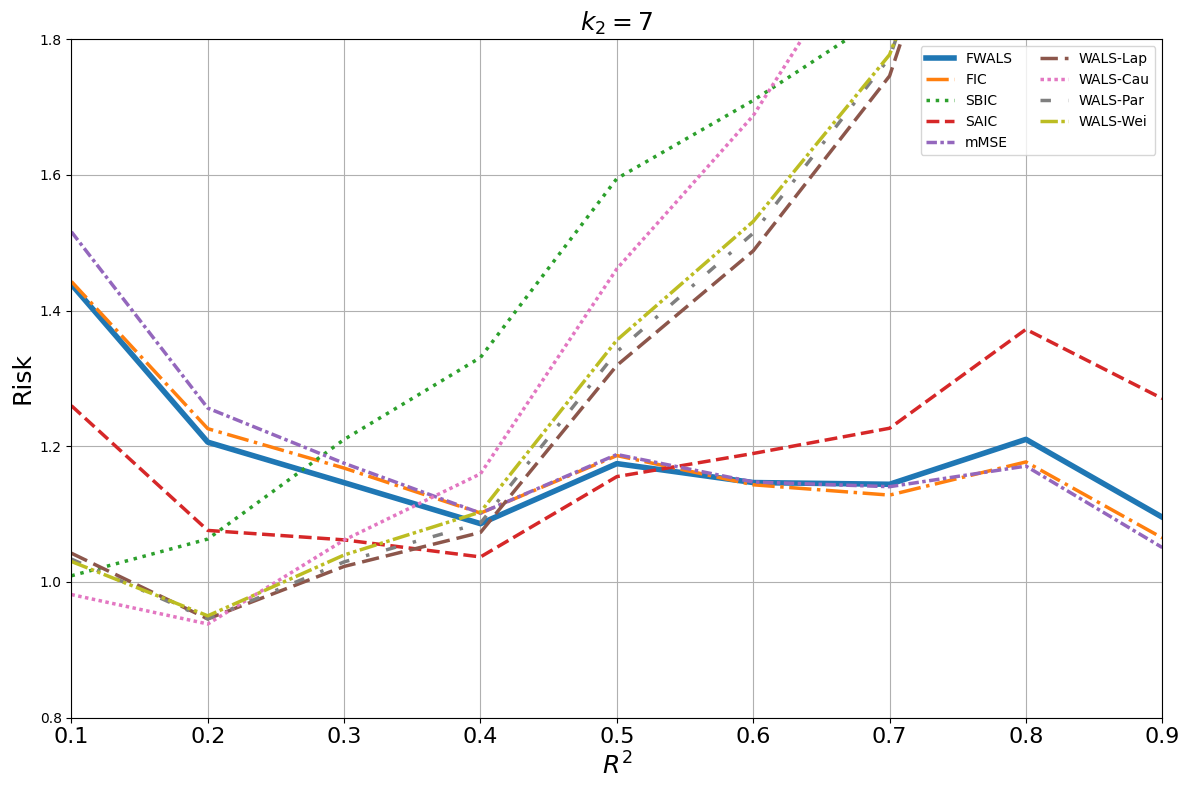}
    \caption{Risk with Different $k_2$s~($N=100$,$\tau=0.7$)}
    \label{fig3}
\end{figure}

\begin{figure}[!htpb]
    \centering
    \includegraphics[width=0.32\linewidth]{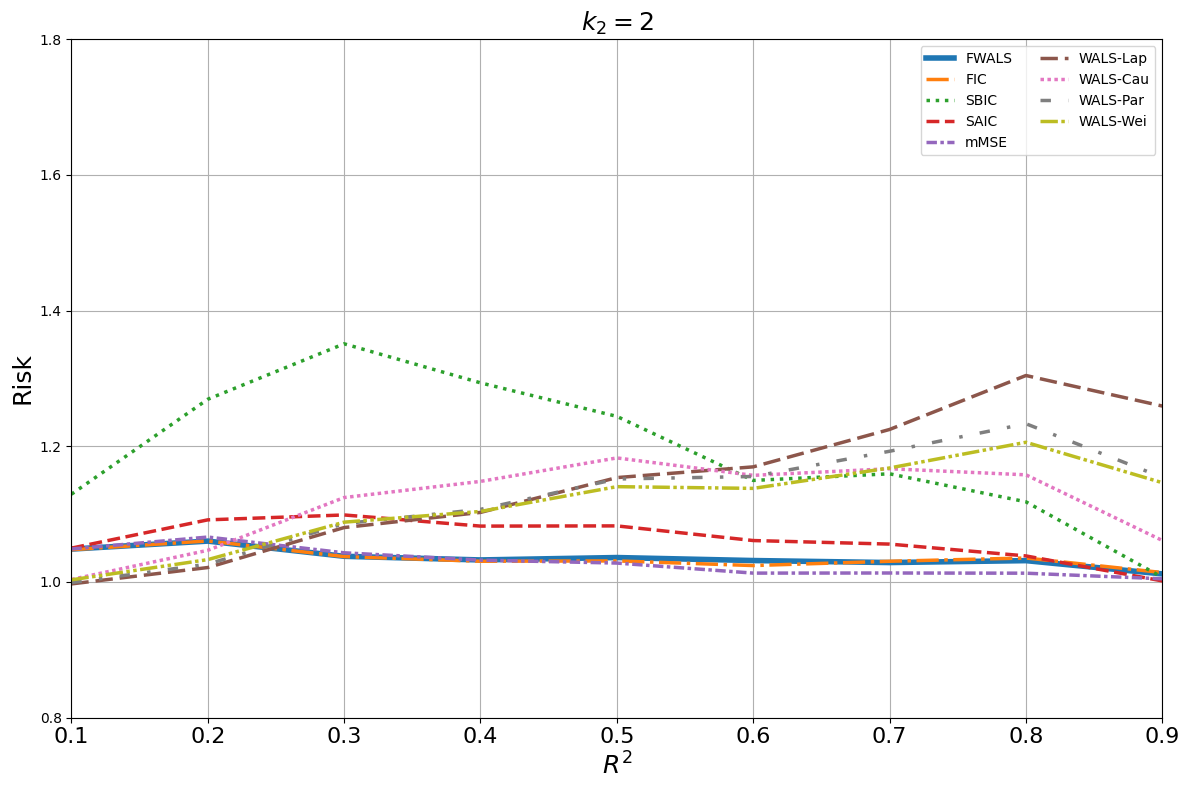}
    \includegraphics[width=0.32\linewidth]{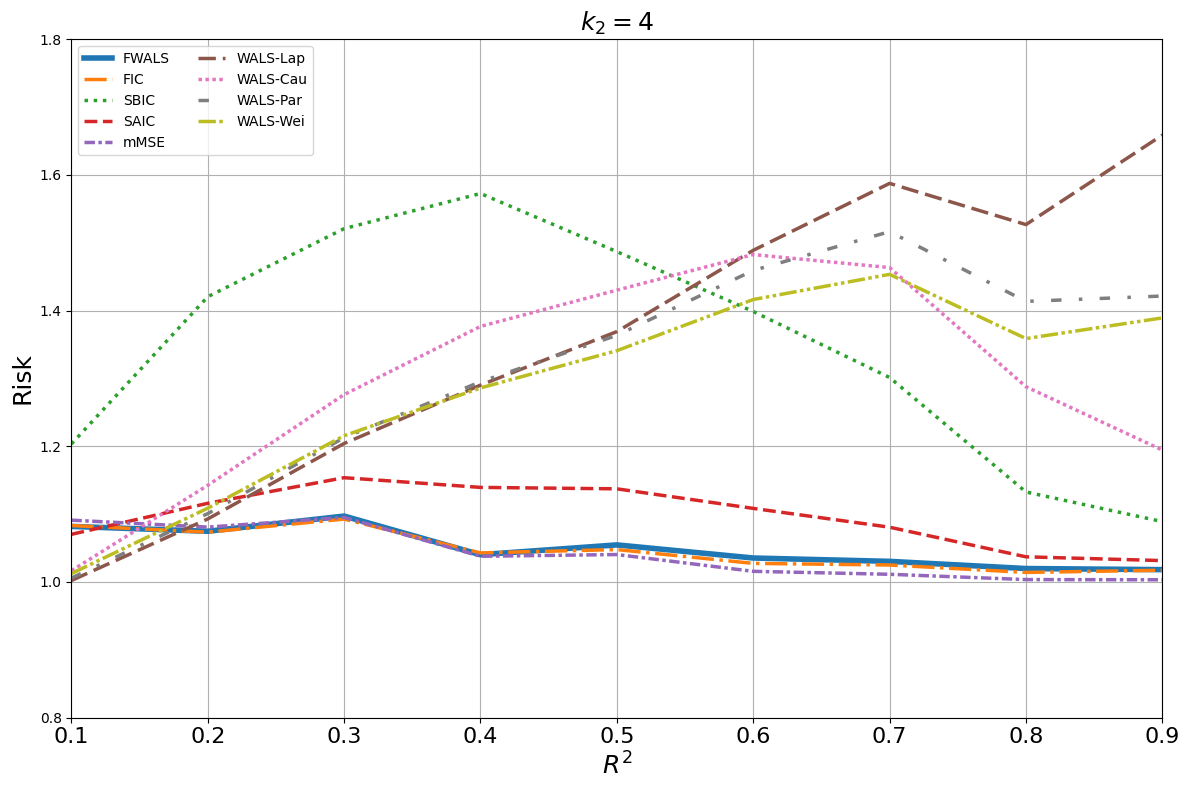}
    \includegraphics[width=0.32\linewidth]{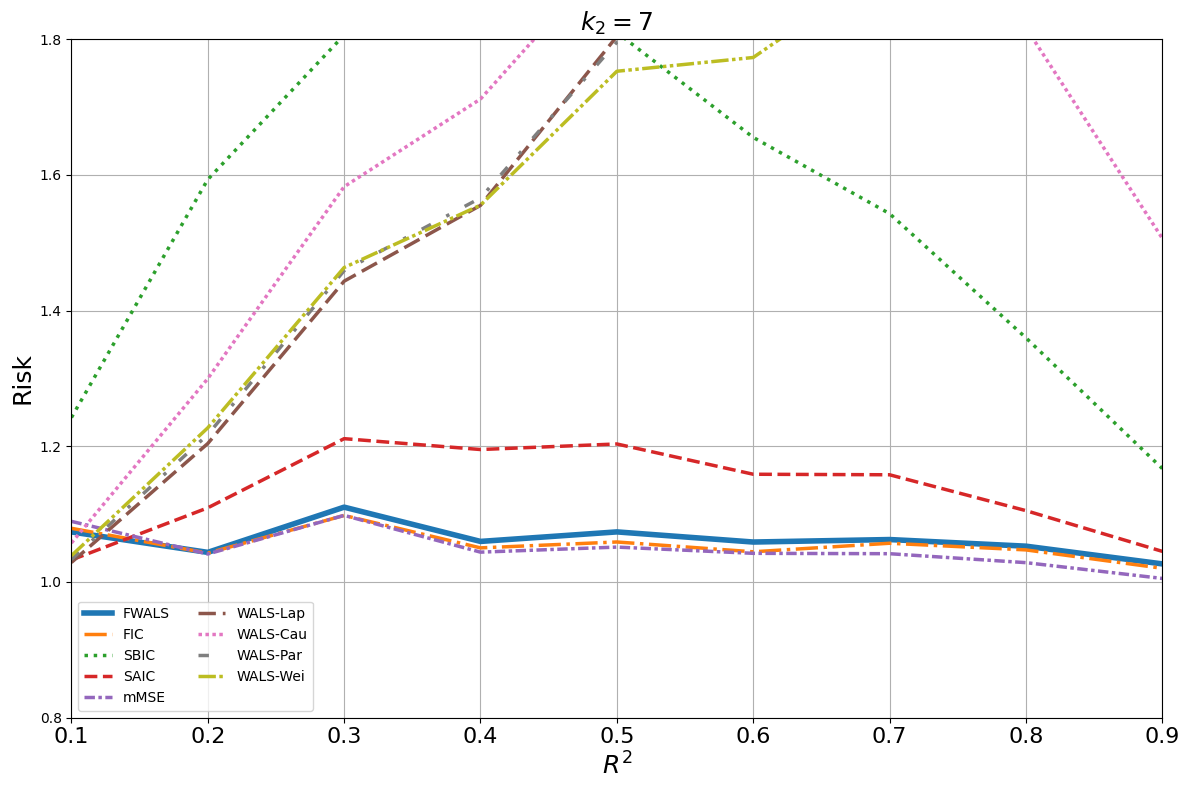}
    \caption{Risk with Different $k_2$s~($N=200$,$\tau=0.3$)}
    \label{fig4}
\end{figure}

\begin{figure}[!htpb]
    \centering
    \includegraphics[width=0.32\linewidth]{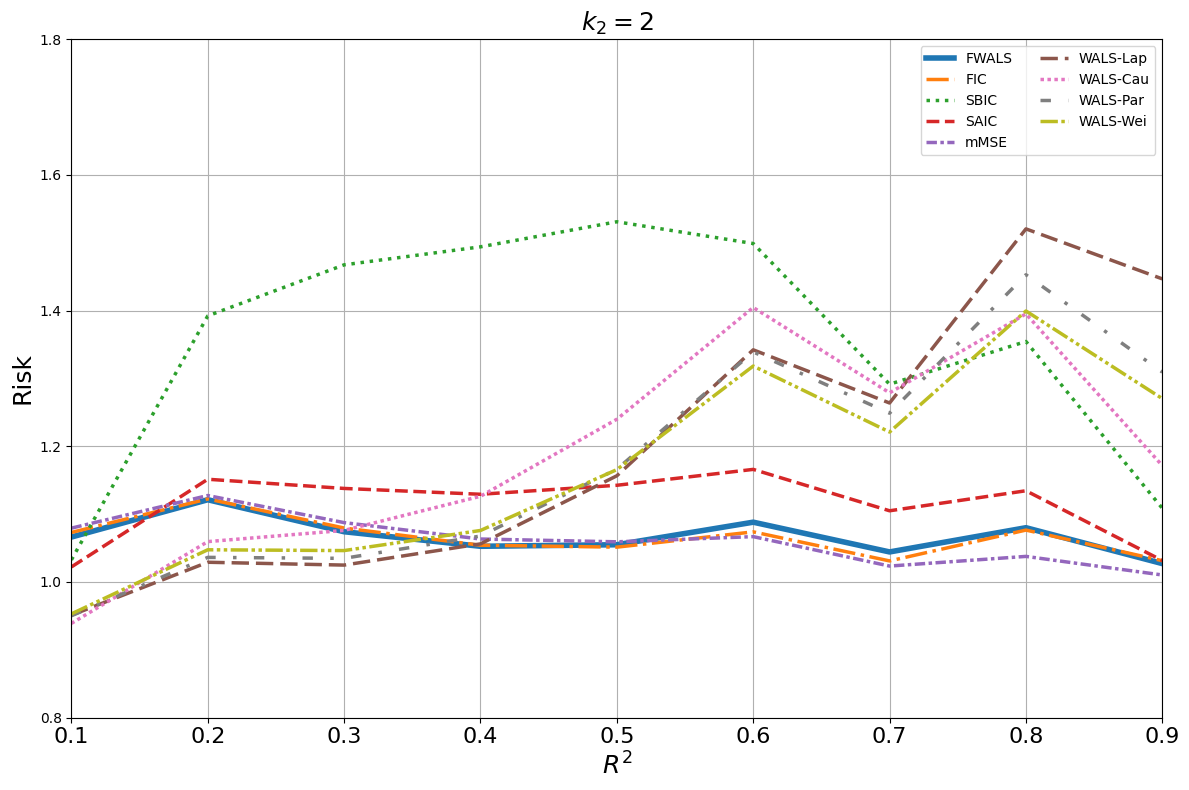}
    \includegraphics[width=0.32\linewidth]{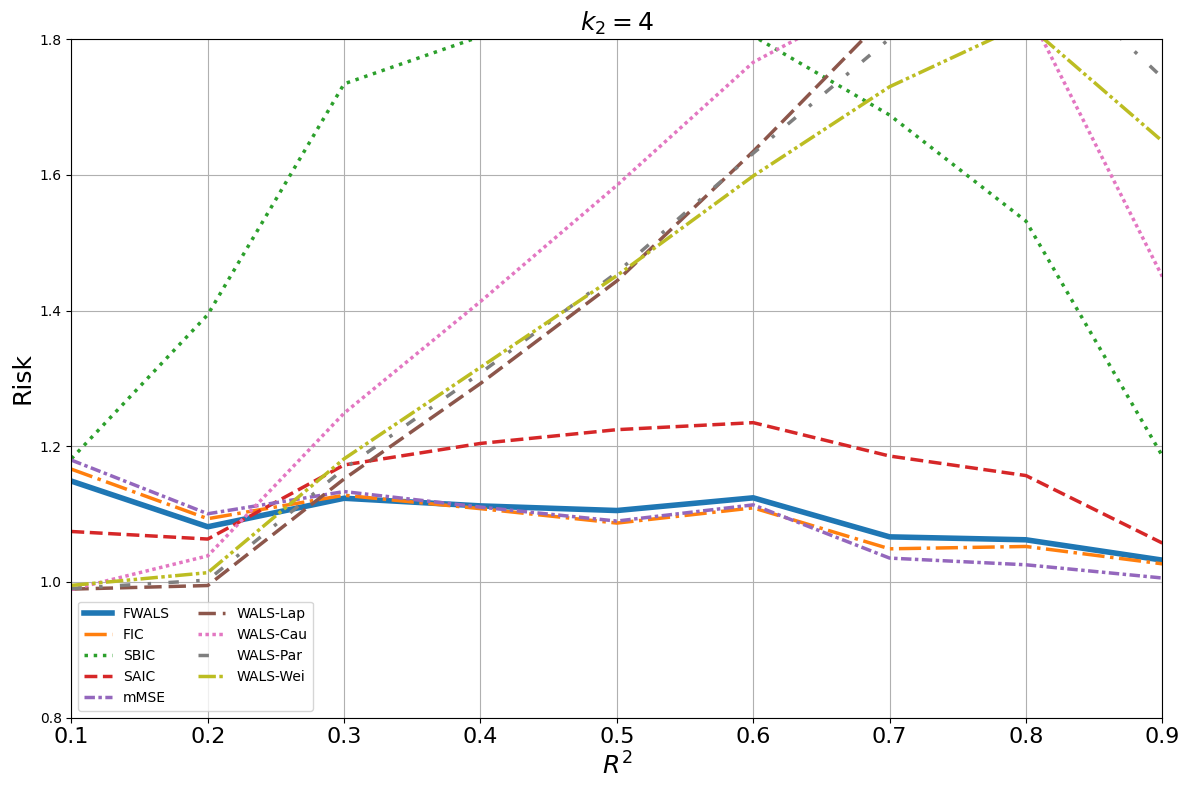}
    \includegraphics[width=0.32\linewidth]{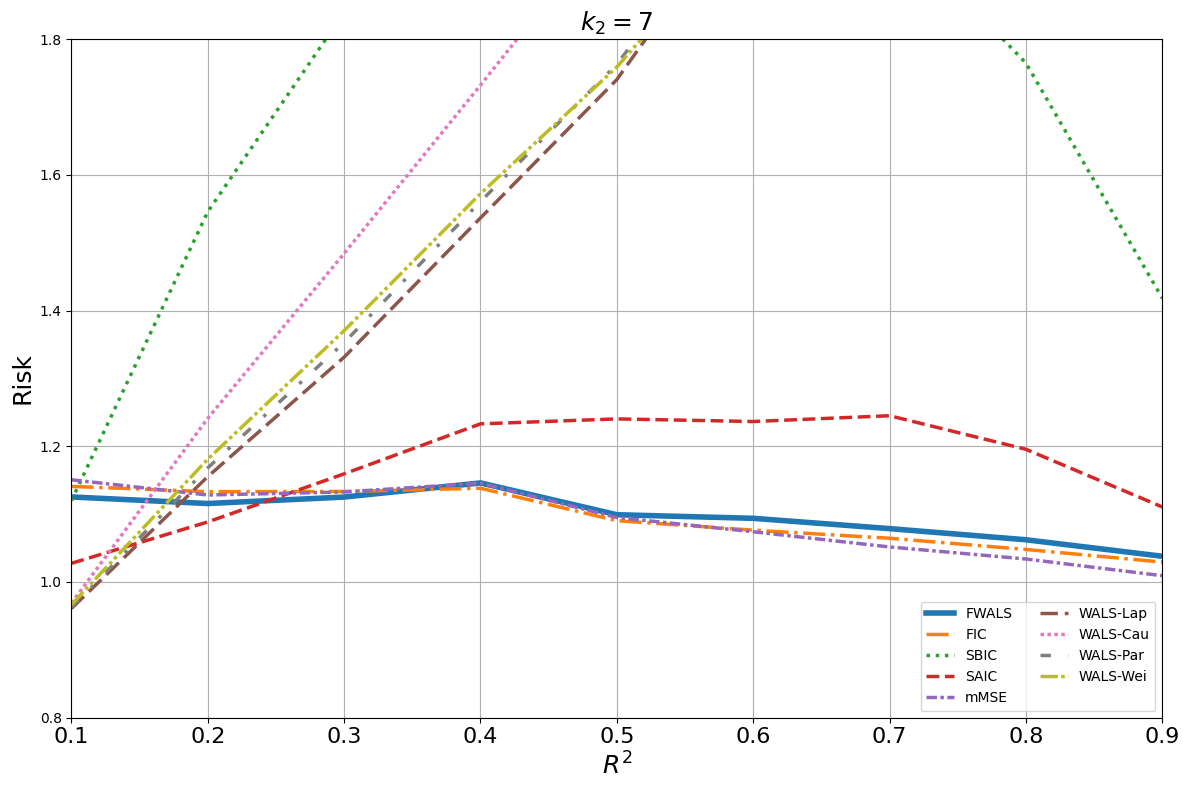}
    \caption{Risk with Different $k_2$s~($N=200$,$\tau=0.5$)}
    \label{fig5}
\end{figure}

\begin{figure}[!htpb]
    \centering
    \includegraphics[width=0.32\linewidth]{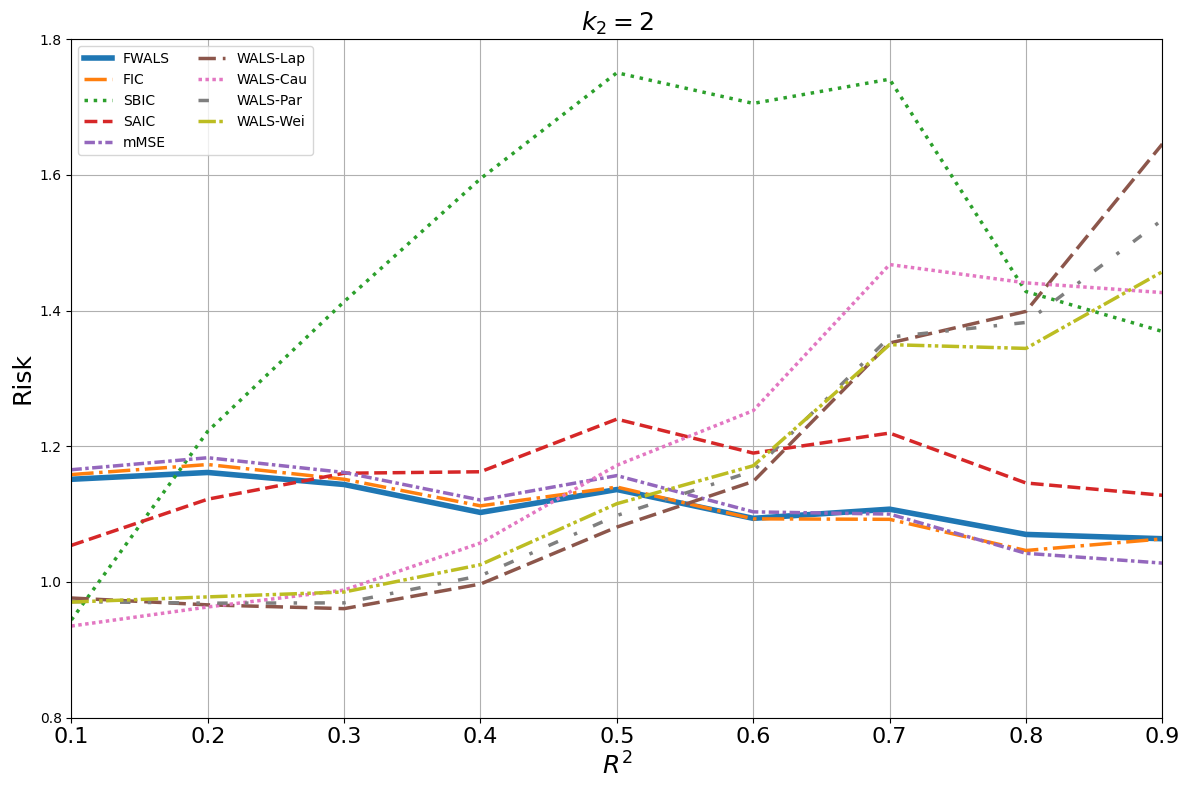}
    \includegraphics[width=0.32\linewidth]{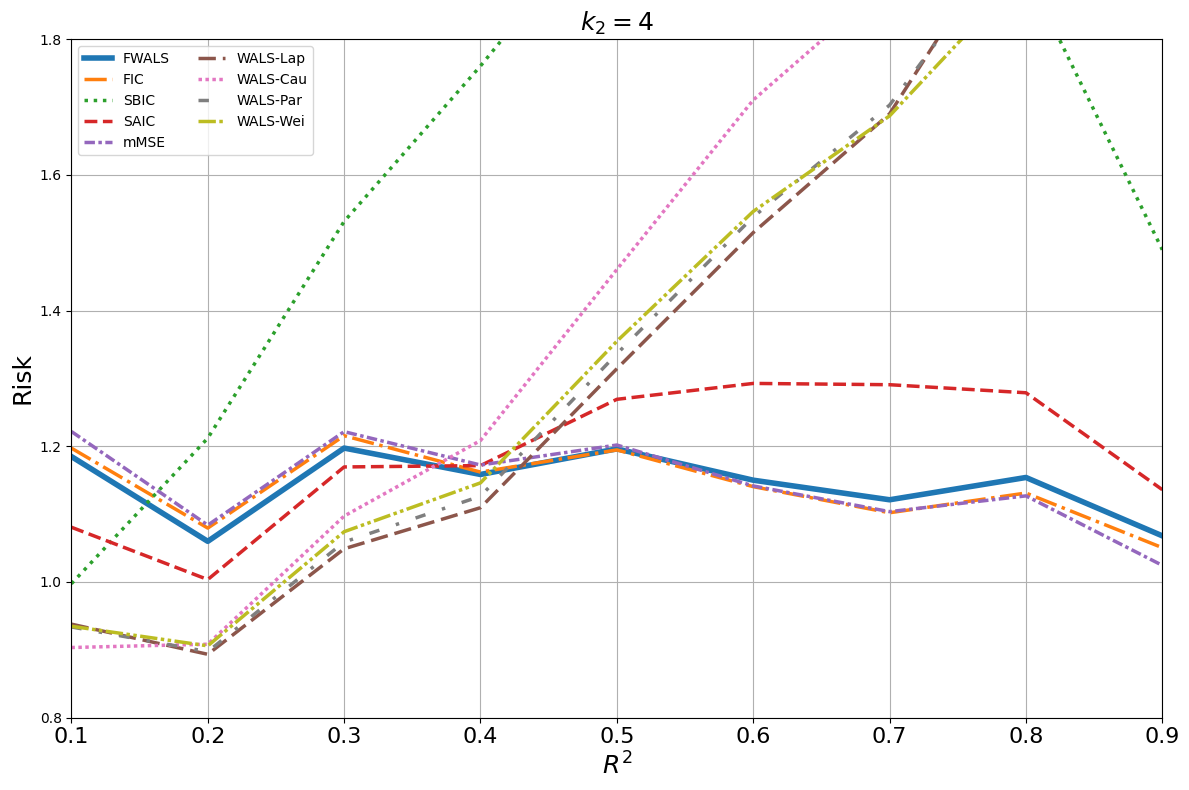}
    \includegraphics[width=0.32\linewidth]{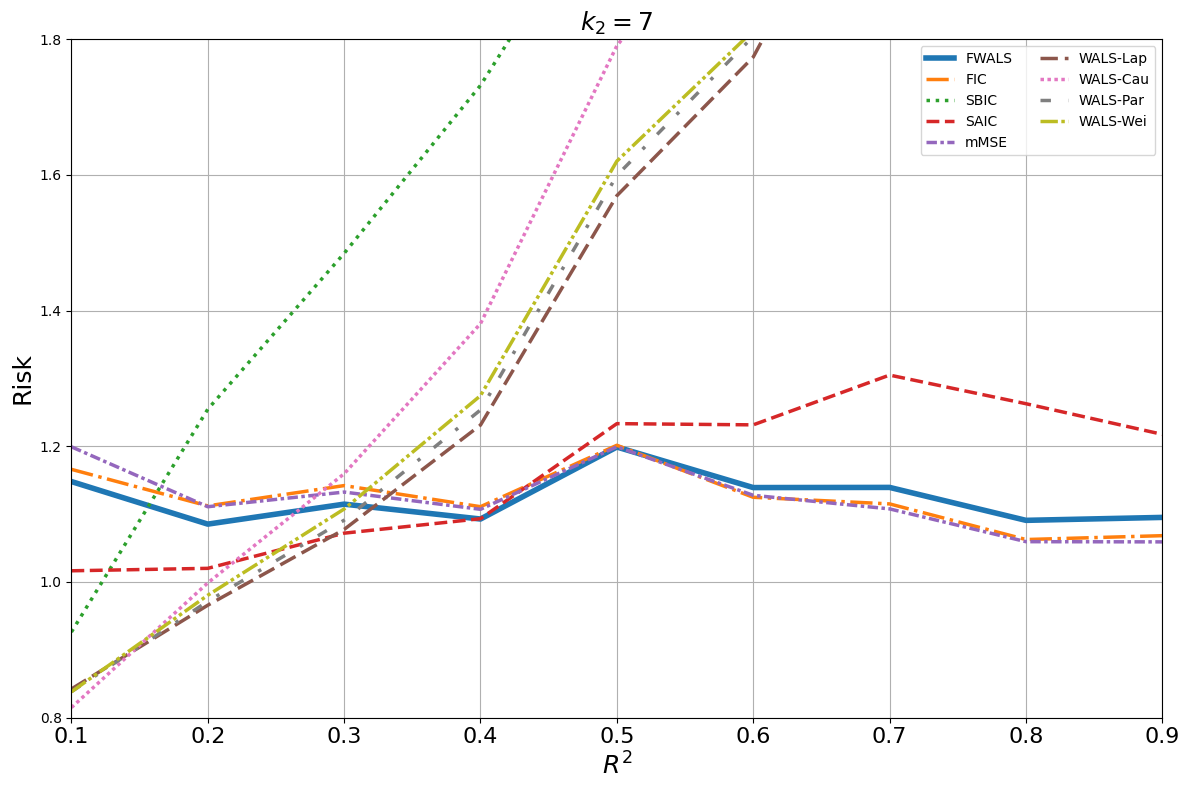}
    \caption{Risk with Different $k_2$s~($N=200$,$\tau=0.7$)}
    \label{fig6}
\end{figure}

\begin{figure}[!htpb]
    \centering
    \includegraphics[width=0.35\linewidth]{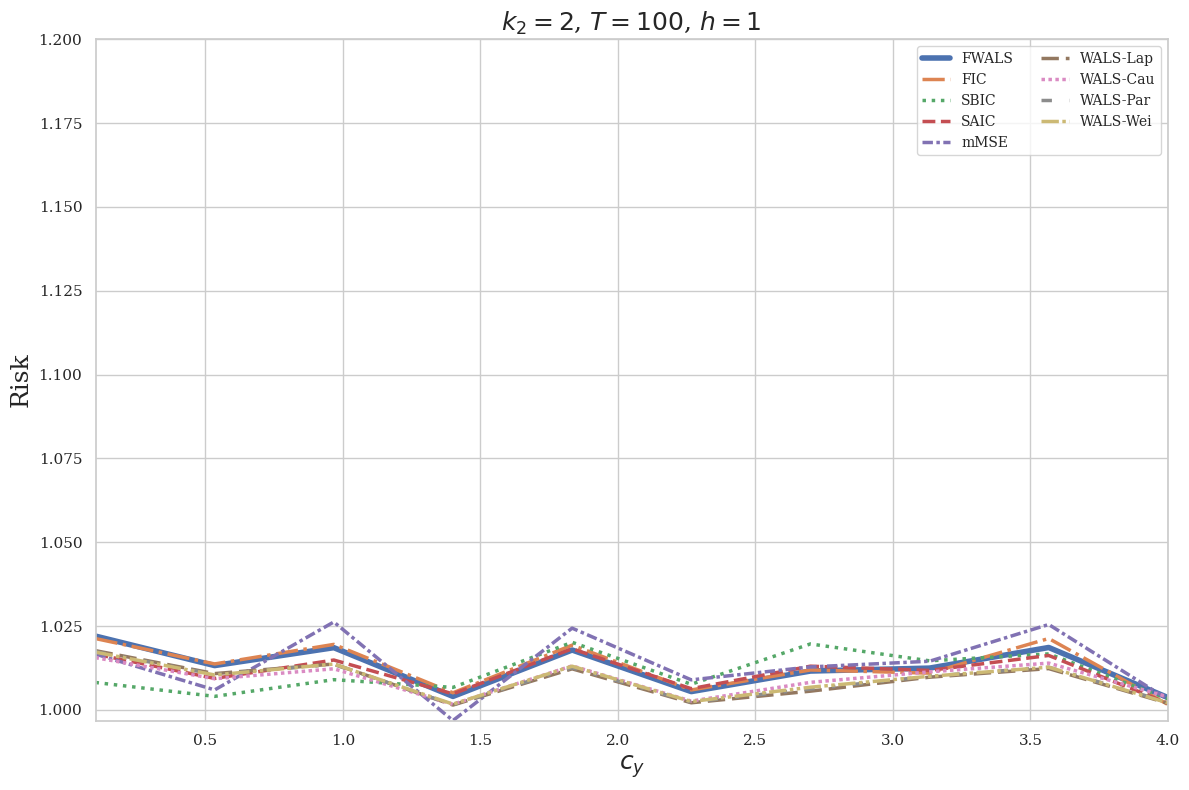}
    \includegraphics[width=0.35\linewidth]{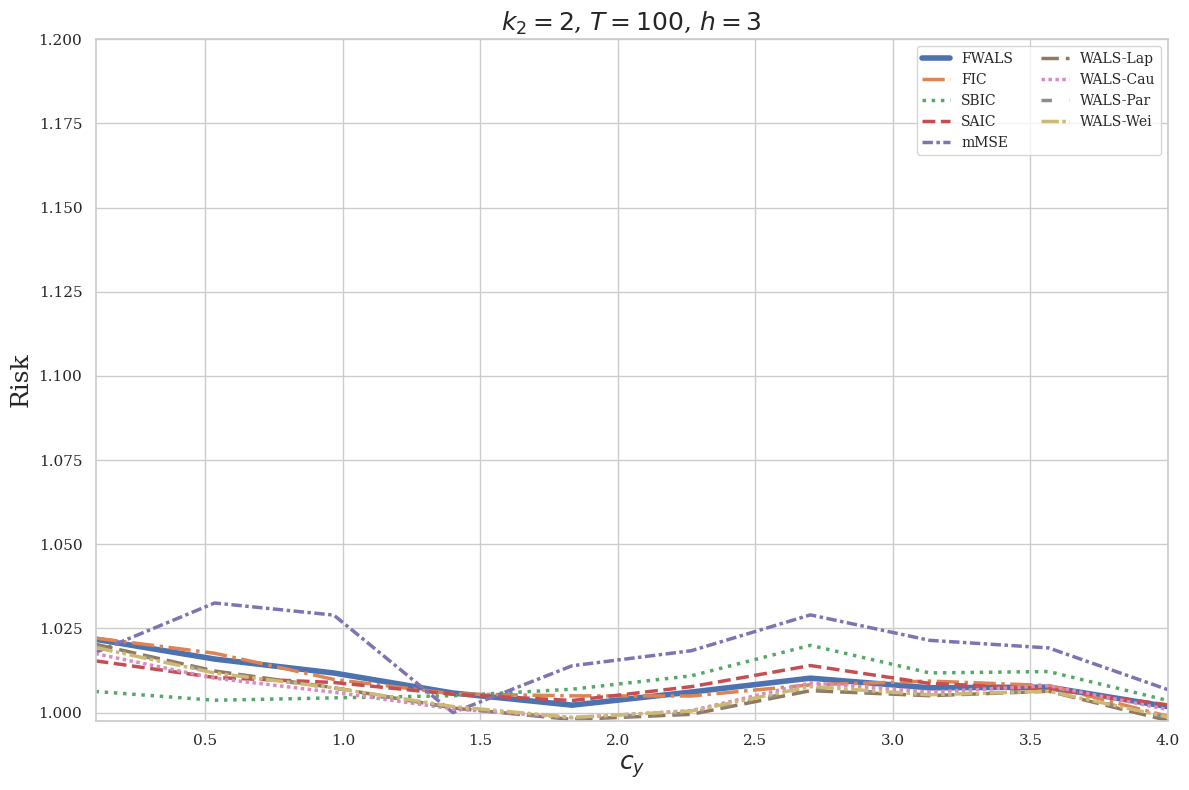}\\
    \includegraphics[width=0.35\linewidth]{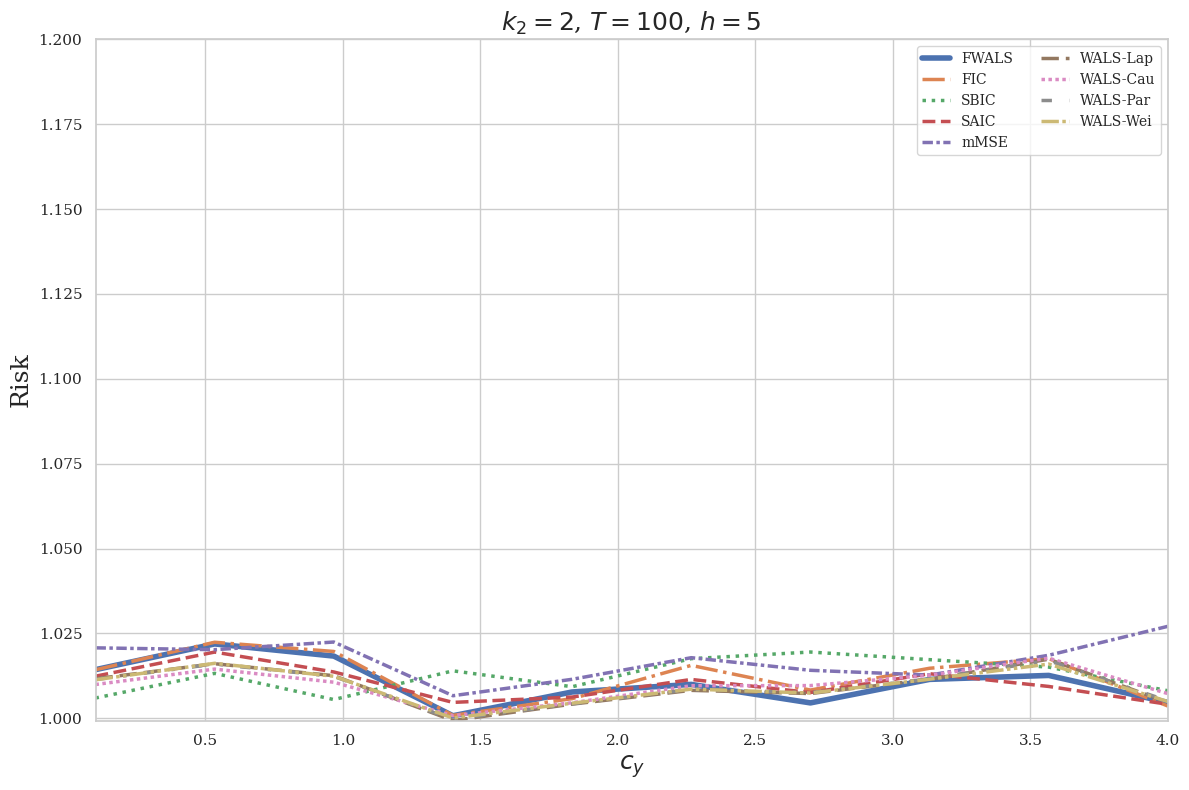}
    \includegraphics[width=0.35\linewidth]{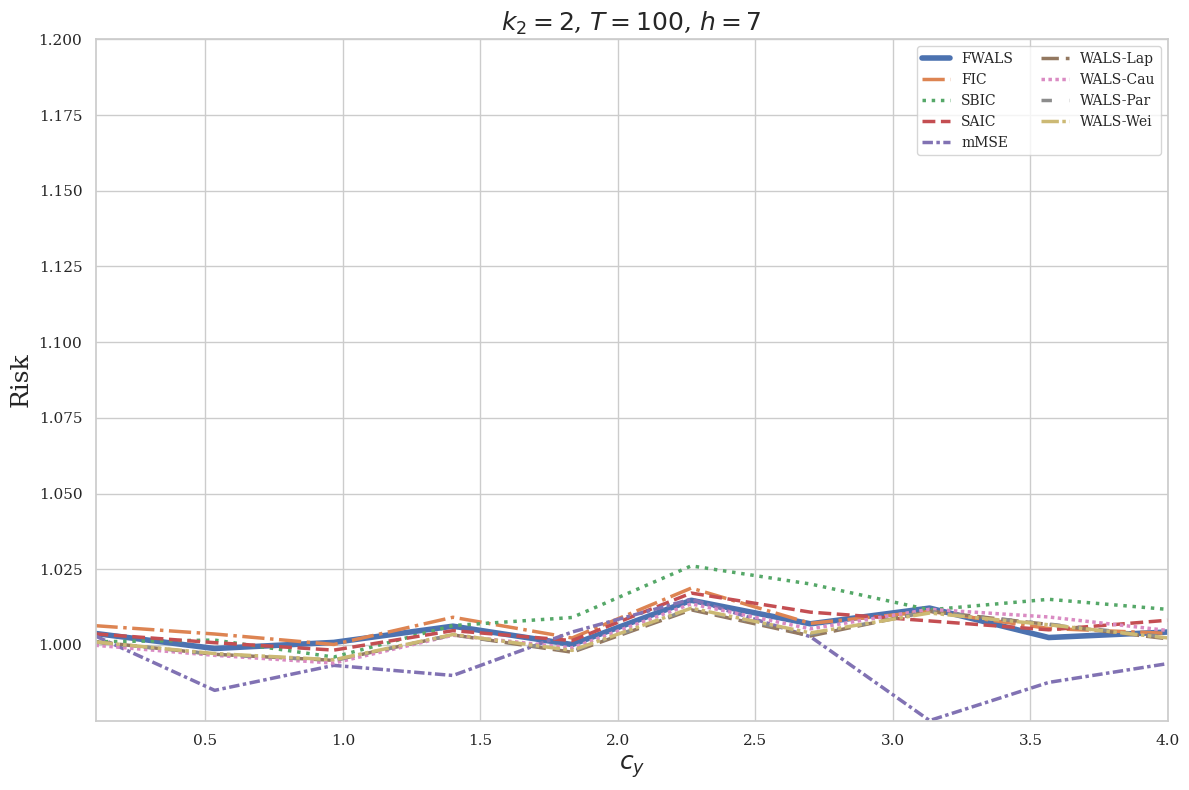}
    \caption{Risk with Different $h$s~($T=100$,$k_2=2$)}
    \label{fig7}
\end{figure}

\begin{figure}[!htpb]
    \centering
    \includegraphics[width=0.35\linewidth]{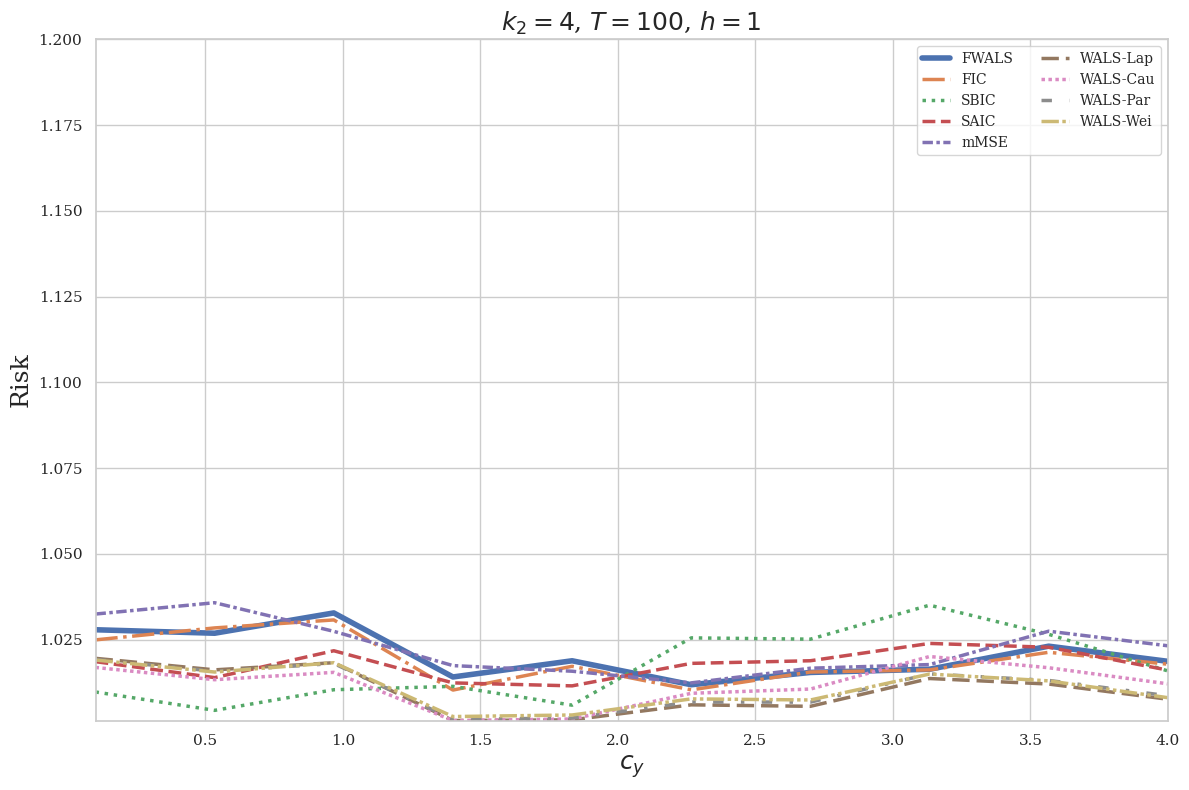}
    \includegraphics[width=0.35\linewidth]{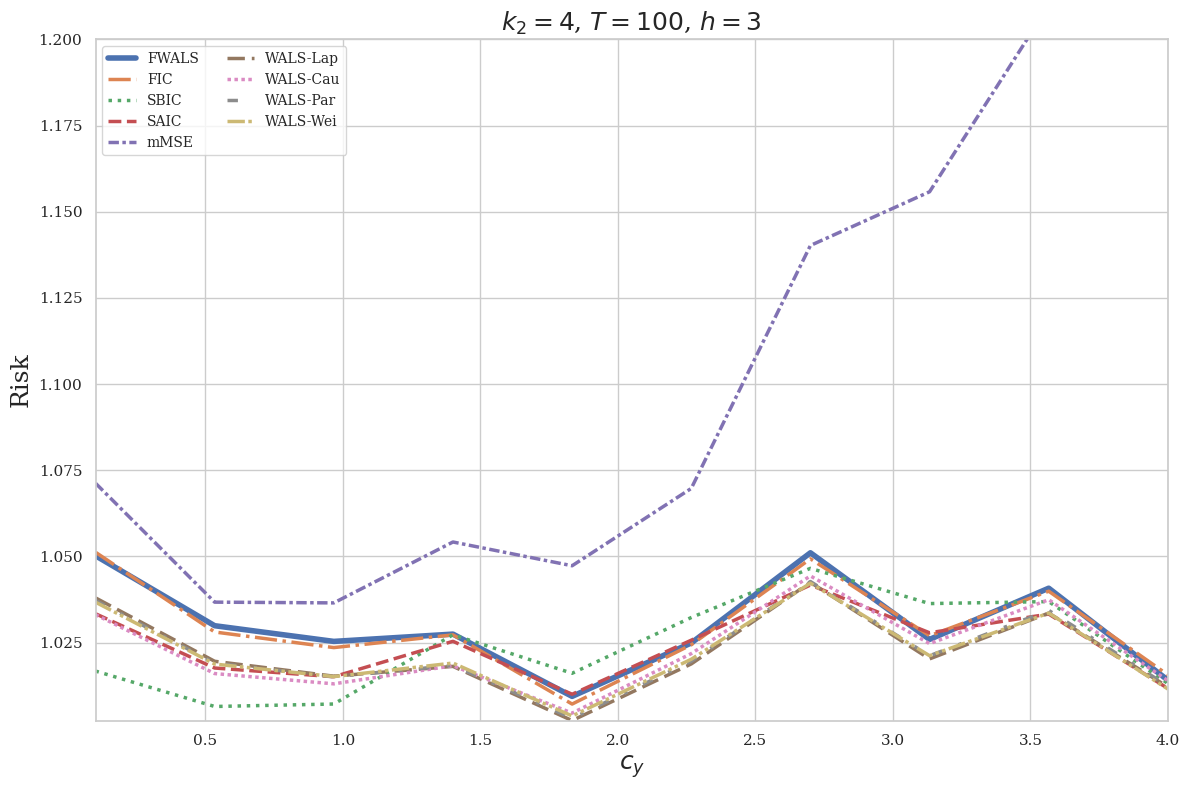}\\
    \includegraphics[width=0.35\linewidth]{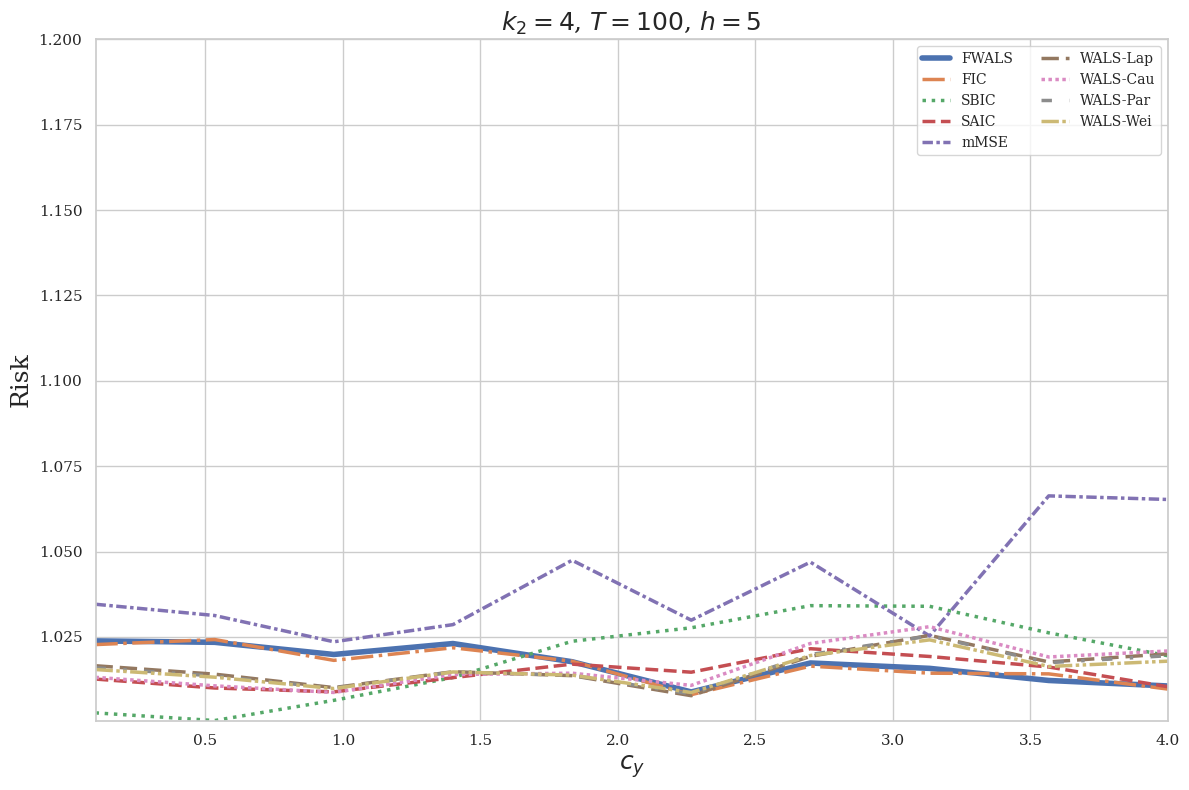}
    \includegraphics[width=0.35\linewidth]{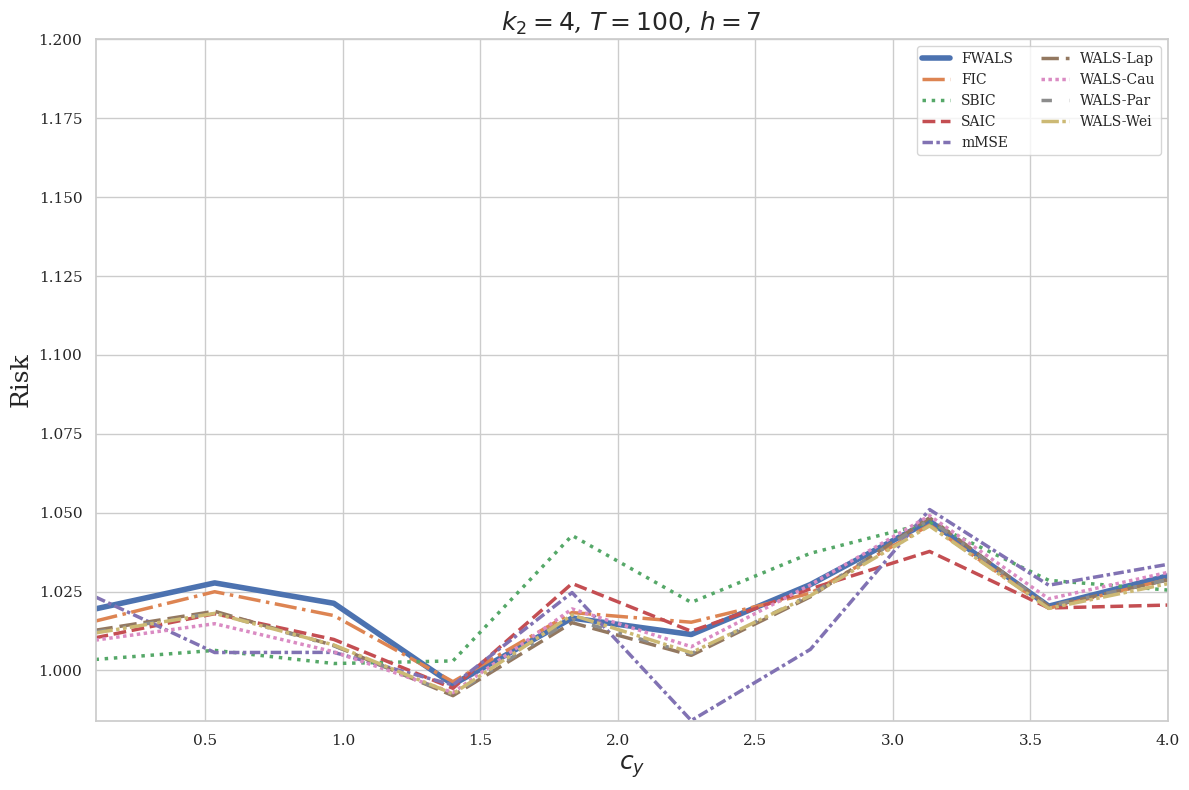}
    \caption{Risk with Different $h$s~($T=100$,$k_2=4$)}
    \label{fig8}
\end{figure}

\begin{figure}[!htpb]
    \centering
    \includegraphics[width=0.35\linewidth]{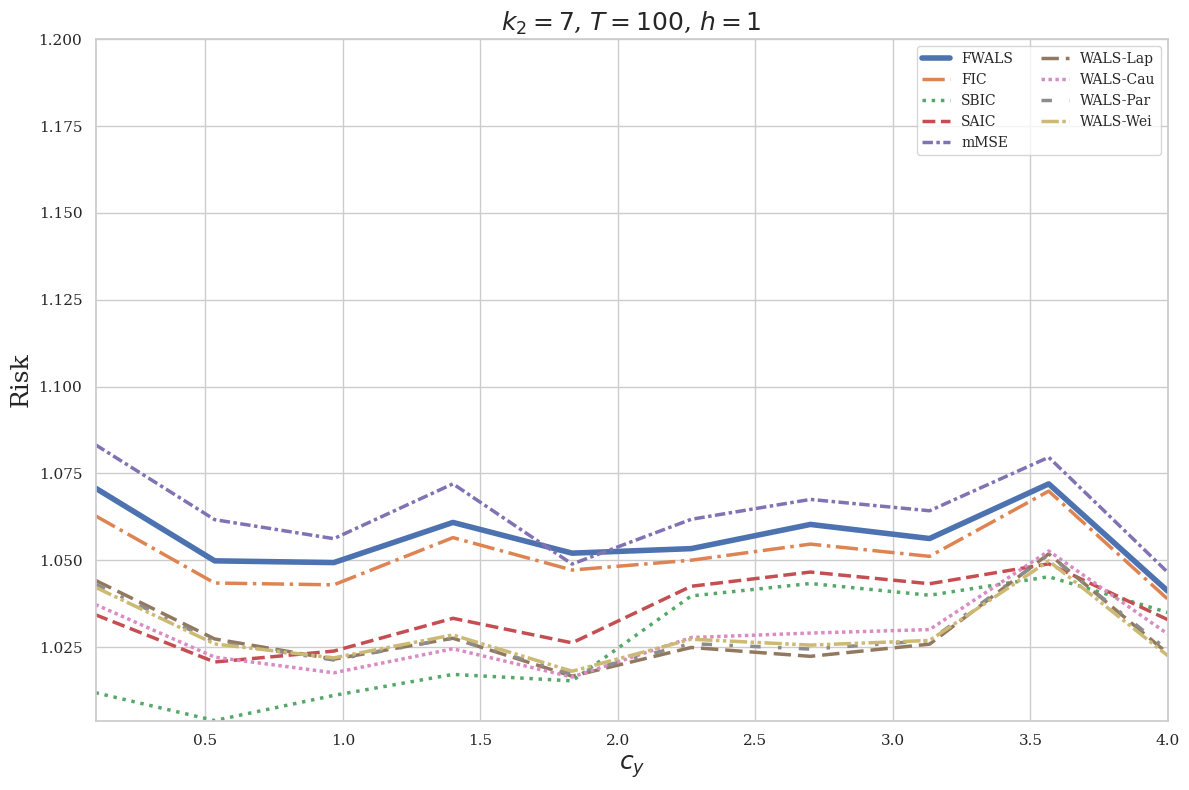}
    \includegraphics[width=0.35\linewidth]{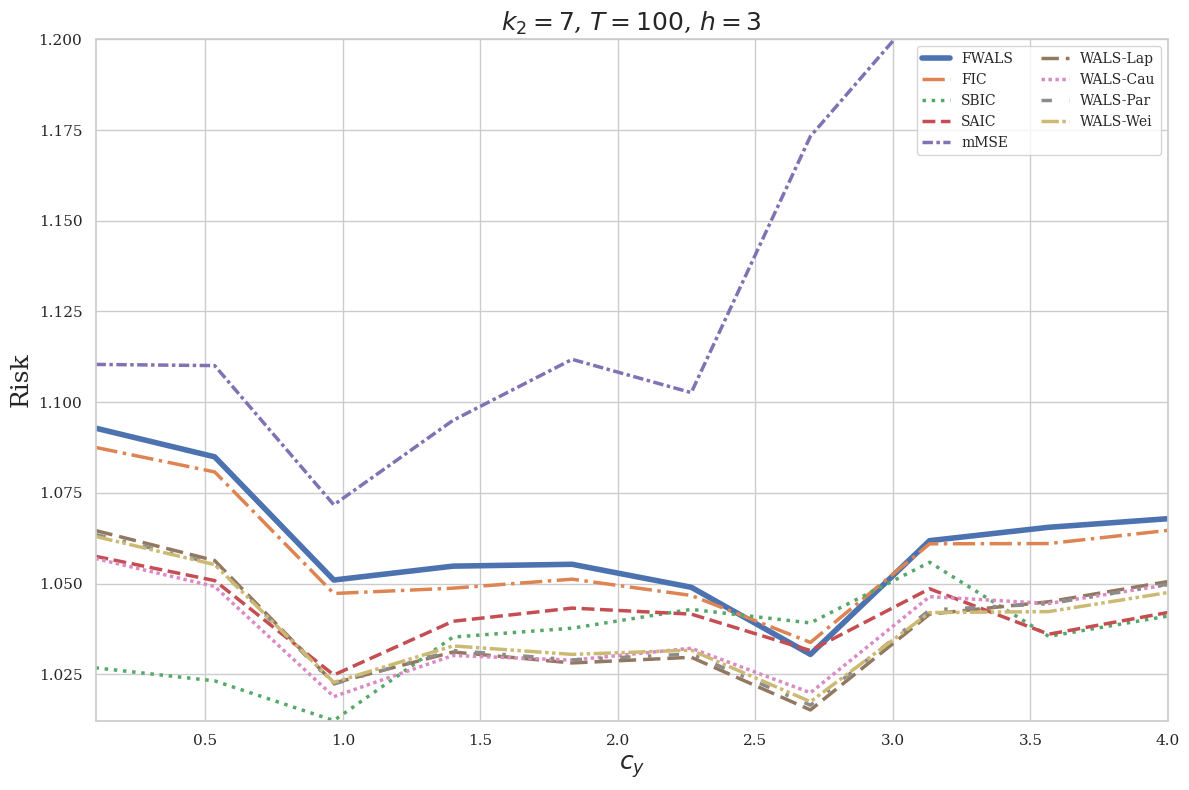}\\
    \includegraphics[width=0.35\linewidth]{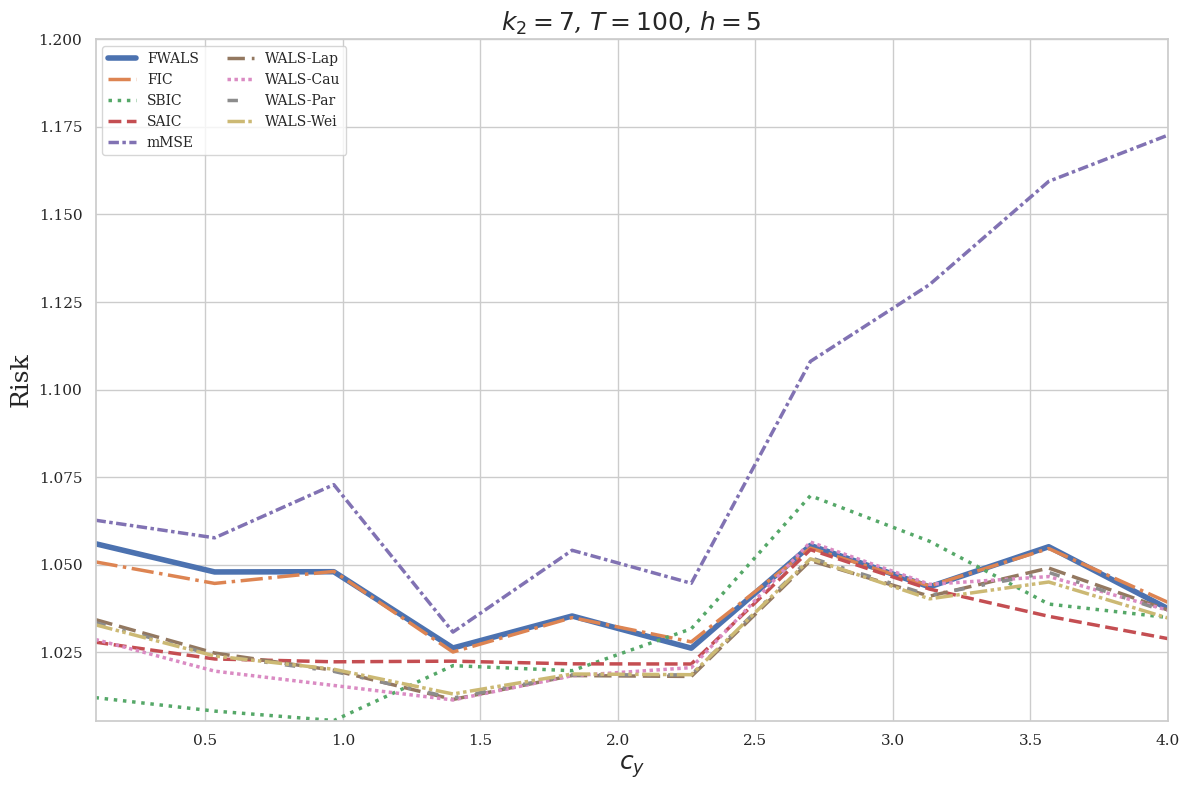}
    \includegraphics[width=0.35\linewidth]{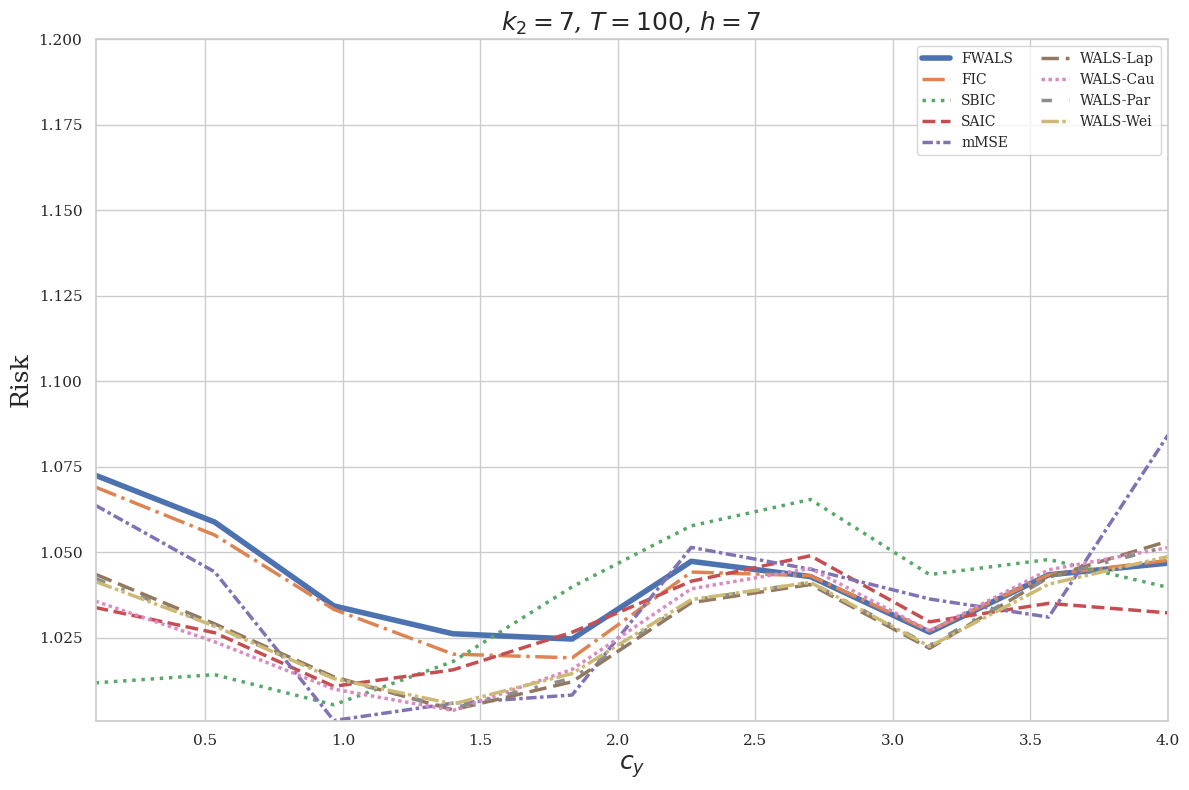}
    \caption{Risk with Different $h$s~($T=100$,$k_2=7$)}
    \label{fig9}
\end{figure}

\subsection{Computational Time Comparison}
To highlight the computational advantage of the proposed approach, we document the actual computational time required to obtain the weights and estimates. As pointed out in the theoretical sections, traditional frequentist model averaging techniques require estimating and combining all possible $2^{k_2}$ sub-models. In contrast, our proposed approach relies on the semi-orthogonal transformation, which effectively reduces the weighting problem from $2^{k_2}$ sub-models down to $k_2$ orthogonalized terms. While the computational burden might not be obvious when the number of auxiliary regressors is small (e.g., $k_2 \le 4$), the difference becomes drastically apparent as $k_2$ increases.

To illustrate this, we conduct an experiment using simulated datasets with $N=100$. We compare our proposed approach against the $\mathsf{FIC}$ and $\mathsf{mMSE}$ approaches across moderately large dimensions of auxiliary regressors, $k_2 \in \{8, 9, 10, 11\}$. All computations are executed on a desktop computer operating on Pop!\_OS 22.04~(linux-based), equipped with an Intel\textsuperscript{\tiny\textregistered} Core\textsuperscript{\tiny\texttrademark} i9-10980XE CPU. The reported computational times are calculated as the average over 100 replications across all designs specified in the first experiment.

\begin{table}[!htpb]
\centering
\begin{tabular}{lccc}
\toprule
 & \multicolumn{3}{c}{$N=100$} \\
\cmidrule{2-4}
$k_2$ & WFIC & FIC & mMSE \\
\midrule
8 & 0.012 & 0.851 & 0.130 \\
9 & 0.015 & 10.330 & 0.992 \\
10 & 0.016 & 137.098 & 9.804 \\
11 & 0.017 & 1407.540 & 97.055 \\
\bottomrule
\end{tabular}
\caption{Computational Time (seconds) for Different Methods}
\label{tab:comp_time}
\end{table}

The results are reported in Table \ref{tab:comp_time}. As we can see, when $k_2=8$, the $\mathsf{FIC}$ takes 0.851 seconds and the $\mathsf{mMSE}$ takes 0.130 seconds, whereas our proposed approach finishes in merely 0.012 seconds. When $k_2$ reaches 11, $\mathsf{FIC}$ requires more than 1,400 seconds (over 23 minutes) to compute the weights. In contrast, the computational time for our proposed approach scales linearly and remains exceptionally low at 0.017 seconds. This evidence demonstrates that the proposed approach is scalable and practically useful when $k_2$ is moderate to large, where enumeration over $2^{k_2}$ candidate models is computationally prohibitive. The finite sample performance are similar to those cases when $k_2$ is small and are reported in the Appendix C~(Figures \ref{fig10}--\ref{fig12}).

\section{Concluding Remarks}

This paper develops a computationally efficient framework for focused model averaging by introducing the focused weighted-average least squares (FWALS) estimator based on orthogonalized auxiliary regressors. The proposed method addresses a key limitation of conventional focused model averaging, namely the exponential growth of sub-models as the number of auxiliary regressors increases. By transforming the auxiliary regressors following \citet{Magnus2010} and \citet{de2018weighted}, we reduce the dimensionality of the weighting problem, thereby providing a tractable alternative without sacrificing statistical validity.

In our simulation study, FWALS, FIC, and the minimum-MSE averaging method based on singleton equations perform similarly in baseline designs, confirming that the orthogonalization step does not compromise efficiency. In the impulse response function setting, FWALS and FIC remain closely aligned and deliver stable risk performance across horizons, whereas the singleton-based approach can exhibit increased risk at longer horizons due to the instability induced by negative weights. We also compare FWALS with prior-based WALS methods (Laplace, Cauchy, Pareto, and Weibull priors). These prior-based procedures can be competitive in some configurations, particularly when the auxiliary signals are weak, but they are not constructed to minimize the risk of the focused parameter. In contrast, FWALS selects regressor-wise weights by minimizing a plug-in AMSE criterion designed for the focused parameter, which provides a principled approach to focused inference. Taken together, the results suggest that FWALS is a competitive and computationally attractive alternative to traditional focused averaging, especially when the full enumeration of $2^{k_2}$ candidate models is infeasible.

\newpage
\bibliography{mfic}

\section*{Appendix A}
For notation simplicity, we let $\BXi=\MBQ_{11}^{-1}\MBQ_{12}$ and $\MBC=\BLambda\MBP^{-1/2}$ throughout the proof.
\begin{proof}[Proof of Theorem 1]
We first discuss the limiting behavior of $\hat{\Bbeta}_2$. Recall that $\hat{\Bbeta}_2=\frac{{\MBX_2^*}^{\tp}\MBM_1\MBy}{N}$, we can have the following representation:
\begin{align*}
\hat{\Bbeta}_2 =& \frac{{\MBX_2^*}^{\tp}\MBM_1(\MBX_2^*\Bbeta_2^*+\Bepsilon)}{N}\\
               =&\Bbeta_2^*+\frac{1}{N}\hat{\MBP}^{-1/2}\hat{\BLambda}\MBX_2^{\tp}\left(\MBI-\MBX_1(\MBX_1^{\tp}\MBX_1)^{-1}\MBX_1^{\tp}\right)\Bepsilon\\
               =&\Bbeta_2^*+\frac{1}{N}\hat{\MBP}^{-1/2}\hat{\BLambda}\MBX_2^{\tp}\Bepsilon-\frac{1}{N}\hat{\MBP}^{-1/2}\hat{\BLambda}\MBX_2^{\tp}\MBX_1(\MBX_1^{\tp}\MBX_1)^{-1}\MBX_1^{\tp}\Bepsilon.
\end{align*}
Accordingly, it can be shown that
\begin{align*}
\sqrt{N}\hat{\Bbeta}_2 =& \sqrt{N}\Bbeta_2^*+\hat{\MBP}^{-1/2}\hat{\BLambda}\frac{1}{\sqrt{N}}\MBX_2^{\tp}\Bepsilon-\hat{\MBP}^{-1/2}\hat{\BLambda}\frac{1}{N}\MBX_2^{\tp}\MBX_1\left(\frac{1}{N}\MBX_1^{\tp}\MBX_1\right)^{-1}\frac{1}{\sqrt{N}}\MBX_1^{\tp}\Bepsilon\\
=&\hat{\MBP}^{1/2}\hat{\BLambda}^{-1}\sqrt{N}\Bbeta_2+\left[\MBzero\quad \hat{\MBP}^{-1/2}\hat{\BLambda}\right]\frac{1}{\sqrt{N}}\MBX^{\tp}\Bepsilon+\left[-\hat{\MBP}^{-1/2}\hat{\BLambda}\frac{1}{N}\MBX_2^{\tp}\MBX_1\left(\frac{1}{N}\MBX_1^{\tp}\MBX_1\right)^{-1}\quad\MBzero\right]\frac{1}{\sqrt{N}}\MBX^{\tp}\Bepsilon\\
=&\hat{\MBP}^{1/2}\hat{\BLambda}^{-1}\Bdelta+\left[\MBzero\quad\hat{\MBP}^{-1/2}\hat{\BLambda}\right]\frac{1}{\sqrt{N}}\MBX^{\tp}\Bepsilon+\left[-\hat{\MBP}^{-1/2}\hat{\BLambda}\hat{\MBQ}_{21}\hat{\MBQ}_{11}^{-1}\quad\MBzero\right]\frac{1}{\sqrt{N}}\MBX^{\tp}\Bepsilon\\
=&\hat{\MBP}^{1/2}\hat{\BLambda}^{-1}\Bdelta+\left[-\hat{\MBP}^{-1/2}\hat{\BLambda}\hat{\MBQ}_{21}\hat{\MBQ}_{11}^{-1}\quad\hat{\MBP}^{-1/2}\hat{\BLambda}\right]\frac{1}{\sqrt{N}}\MBX^{\tp}\Bepsilon.
\end{align*}
The second equality holds because of Equation\eqref{eq:ortho} and third equality comes from Assumption \ref{ass:loca}. Using Assumptions \ref{ass:lln} and \ref{ass:clt}, we can obtain the following result:
\begin{align*}
\sqrt{N}\hat{\Bbeta}_2-\hat{\MBP}^{1/2}\hat{\BLambda}^{-1}\Bdelta\CD\left[-\MBC^{\tp}\BXi^{\tp}\quad\MBC^{\tp}\right]\MBR.
\end{align*}

As for $\hat{\Bbeta}_{1\mathrm{WALS}}$, based on Equation \eqref{eq:WALS2} and result from $\hat{\Bbeta}_2$, we can obtain that
\begin{align*}
\hat{\Bbeta}_{1\mathrm{WALS}} =&\Bbeta_1+\left(\frac{\MBX_1^{\tp}\MBX_1}{N}\right)^{-1}\frac{\MBX_1^{\tp}\MBX_2}{N}\Bbeta_2+\left(\frac{\MBX_1^{\tp}\MBX_1}{N}\right)^{-1}\frac{\MBX_1^{\tp}\Bepsilon}{N}\\
&-\left(\frac{\MBX_1^{\tp}\MBX_1}{N}\right)^{-1}\frac{\MBX_1^{\tp}\MBX_2}{N}\hat{\BLambda}\hat{\MBP}^{-1/2}\tilde{\MBW}\left(\hat{\MBP}^{1/2}\hat{\BLambda}^{-1}\Bbeta_2+\left[-\hat{\MBP}^{-1/2}\hat{\BLambda}\hat{\MBQ}_{21}\hat{\MBQ}_{11}^{-1}\quad\hat{\MBP}^{-1/2}\hat{\BLambda}\right]\frac{1}{N}\MBX^{\tp}\Bepsilon\right)\\
=&\Bbeta_1+\hat{\MBQ}_{11}^{-1}\hat{\MBQ}_{12}\Bbeta_2+\hat{\MBQ}_{11}^{-1}\frac{\MBX_1^{\tp}\Bepsilon}{N}\\
&-\hat{\MBQ}_{11}^{-1}\hat{\MBQ}_{12}\hat{\BLambda}\hat{\MBP}^{-1/2}\tilde{\MBW}\left(\hat{\MBP}^{1/2}\hat{\BLambda}^{-1}\Bbeta_2+\left[-\hat{\MBP}^{-1/2}\hat{\BLambda}\hat{\MBQ}_{21}\hat{\MBQ}_{11}^{-1}\quad\hat{\MBP}^{-1/2}\hat{\BLambda}\right]\frac{1}{N}\MBX^{\tp}\Bepsilon\right)\\
=&\Bbeta_1+\hat{\MBQ}_{11}^{-1}\hat{\MBQ}_{12}\left(\MBI-\hat{\BLambda}\hat{\MBP}^{-1/2}\tilde{\MBW}\hat{\MBP}^{1/2}\hat{\BLambda}^{-1}\right)\Bbeta_2\\
&+\left[\hat{\MBQ}_{11}^{-1}+\hat{\MBQ}_{11}^{-1}\hat{\MBQ}_{12}\hat{\BLambda}\hat{\MBP}^{-1/2}\tilde{\MBW}\hat{\MBP}^{-1/2}\hat{\BLambda}\hat{\MBQ}_{21}\hat{\MBQ}_{11}^{-1}\quad-\hat{\MBQ}_{11}^{-1}\hat{\MBQ}_{12}\hat{\BLambda}\hat{\MBP}^{-1/2}\tilde{\MBW}\hat{\MBP}^{-1/2}\hat{\BLambda}\right]\frac{1}{N}\MBX^{\tp}\Bepsilon.
\end{align*}
Consequently, 
\begin{align*}
\sqrt{N}&(\hat{\Bbeta}_{1\mathrm{WALS}}-\Bbeta_1)=\hat{\MBQ}_{11}^{-1}\hat{\MBQ}_{12}\left(\MBI-\hat{\BLambda}\hat{\MBP}^{-1/2}\tilde{\MBW}\hat{\MBP}^{1/2}\hat{\BLambda}^{-1}\right)\Bdelta\\
&+\left[\hat{\MBQ}_{11}^{-1}+\hat{\MBQ}_{11}^{-1}\hat{\MBQ}_{12}\hat{\BLambda}\hat{\MBP}^{-1/2}\tilde{\MBW}\hat{\MBP}^{-1/2}\hat{\BLambda}\hat{\MBQ}_{21}\hat{\MBQ}_{11}^{-1}\quad-\hat{\MBQ}_{11}^{-1}\hat{\MBQ}_{12}\hat{\BLambda}\hat{\MBP}^{-1/2}\tilde{\MBW}\hat{\MBP}^{-1/2}\hat{\BLambda}\right]\frac{1}{\sqrt{N}}\MBX^{\tp}\Bepsilon\\
\CD& \MBQ_{11}^{-1}\MBQ_{12}\left(\MBI-\BLambda\MBP^{-1/2}\tilde{\MBW}{\MBP^{1/2}}^{\tp}\BLambda^{-1}\right)\Bdelta\\
&+\left[\MBQ_{11}^{-1}+\MBQ_{11}^{-1}\MBQ_{12}\BLambda\MBP^{-1/2}\tilde{\MBW}{\MBP^{-1/2}}^{\tp}\BLambda\MBQ_{21}\MBQ_{11}^{-1}\quad-\MBQ_{11}^{-1}\MBQ_{12}\BLambda\MBP^{-1/2}\tilde{\MBW}{\MBP^{-1/2}}^{\tp}\BLambda\right]\MBR\\
=& \BXi\MBC\left(\MBI-\tilde{\MBW}\right)\MBC^{-1}\Bdelta+\left[\MBQ_{11}^{-1}+\BXi\MBC\tilde{\MBW}\MBC^{\tp}\BXi^{\tp}\quad-\BXi\MBC\tilde{\MBW}\MBC^{\tp}\right]\MBR.
\end{align*}
\end{proof}

\begin{proof}[Proof of Theorem 2 and Equation 14]
In this proof, we derive the limiting behavior of the focused parameter. Because the focused parameter is a function of $\Bbeta_1$, and let $\frac{\partial\mu}{\partial\Bbeta_1}=\MBD_{\Bbeta_1}$ , it is easy to adopt the delta method to have
\begin{align*}
\sqrt{N}&\left(\mu(\hat{\Bbeta}_{1\mathrm{WALS}})-\mu(\Bbeta_1)\right) \CD \MBD_{\Bbeta_1}^{\tp}\BXi\MBC\left(\MBI-\tilde{\MBW}\right)\MBC^{-1}\Bdelta\\
&+\MBD_{\Bbeta_1}^{\tp}\left[\MBQ_{11}^{-1}+\BXi\MBC\tilde{\MBW}\MBC^{\tp}\BXi^{\tp}\quad-\BXi\MBC\tilde{\MBW}\MBC^{\tp}\right]\MBR\equiv R_{\mu}.
\end{align*}

Based on the above result, it can be seen that the mean of $R_{\mu}$ follows that
\begin{align*}
\mathbb{E}(R_{\mu}) = \MBD_{\Bbeta_1}^{\tp}\BXi\MBC\left(\MBI-\tilde{\MBW}\right)\MBC^{-1}\Bdelta.
\end{align*}
The result holds because $\mathbb{E}(\MBR)=\MBzero$ from Assumption \ref{ass:clt}. The variance can be derived immediately with a suitable partition that $\BOmega=\begin{bmatrix}\BOmega_{11}&\BOmega_{12}\\
\BOmega_{21}&\BOmega_{22}\end{bmatrix}$.
\begin{align*}
&\VAR(R_{\mu})=\MBD_{\Bbeta_1}^{\tp}\left[\MBQ_{11}^{-1}+\BXi\MBC\tilde{\MBW}\MBC^{\tp}\BXi^{\tp}\quad-\BXi\MBC\tilde{\MBW}\MBC^{\tp}\right]\BOmega\left[\MBQ_{11}^{-1}+\BXi\MBC\tilde{\MBW}\MBC^{\tp}\BXi^{\tp}\quad-\BXi\MBC\tilde{\MBW}\MBC^{\tp}\right]^{\tp}\MBD_{\Bbeta_1}\\
=&\MBD_{\Bbeta_1}^{\tp}\left[\MBQ_{11}^{-1}+\BXi\MBC\tilde{\MBW}\MBC^{\tp}\BXi^{\tp}\quad-\BXi\MBC\tilde{\MBW}\MBC^{\tp}\right]\begin{bmatrix}\BOmega_{11}&\BOmega_{12}\\
\BOmega_{21}&\BOmega_{22}\end{bmatrix}\left[\MBQ_{11}^{-1}+\BXi\MBC\tilde{\MBW}\MBC^{\tp}\BXi^{\tp}\quad-\BXi\MBC\tilde{\MBW}\MBC^{\tp}\right]^{\tp}\MBD_{\Bbeta_1}\\
=&\MBD_{\Bbeta_1}^{\tp}(\MBQ_{11}^{-1}+\BXi\MBC\tilde{\MBW}\MBC^{\tp}\BXi^{\tp})\BOmega_{11}(\MBQ_{11}^{-1}+\BXi\MBC\tilde{\MBW}\MBC^{\tp}\BXi^{\tp})\MBD_{\Bbeta_1}\\
&-\MBD_{\Bbeta_1}^{\tp}(\MBQ_{11}^{-1}+\BXi\MBC\tilde{\MBW}\MBC^{\tp}\BXi^{\tp})\BOmega_{12}\MBC\tilde{\MBW}\MBC^{\tp}\BXi^{\tp}\MBD_{\Bbeta_1}\\
&-\MBD_{\Bbeta_1}^{\tp}\BXi\MBC\tilde{\MBW}\MBC^{\tp}\BOmega_{12}(\MBQ_{11}^{-1}+\BXi\MBC\tilde{\MBW}\MBC^{\tp}\BXi^{\tp})\MBD_{\Bbeta_1}\\
&+\MBD_{\Bbeta_1}^{\tp}\BXi\MBC\tilde{\MBW}\MBC^{\tp}\BOmega_{22}\MBC\tilde{\MBW}\MBC^{\tp}\BXi^{\tp}\MBD_{\Bbeta_1}\\
=& \MBD_{\Bbeta_1}^{\tp}\MBQ_{11}^{-1}\MBD_{\Bbeta_1}+\tilde{\MBw}^{\tp}\MBV\MBC^{\tp}\BXi^{\tp}\BOmega_{11}\BXi\MBC\MBV\tilde{\MBw}-\tilde{\MBw}^{\tp}\MBV\MBC^{\tp}\BXi^{\tp}\BOmega_{12}\MBC\MBV\tilde{\MBw}-\tilde{\MBw}^{\tp}\MBV\MBC^{\tp}\BOmega_{21}\BXi\MBC\MBV\tilde{\MBw}\\
&+\tilde{\MBw}^{\tp}\MBV\MBC^{\tp}\BOmega_{22}\MBC\MBV\tilde{\MBw}+2\tilde{\MBw}^{\tp}\MBV\MBC^{\tp}\BXi^{\tp}\BOmega_{11}\MBQ_{11}^{-1}\MBD_{\Bbeta_1}-2\tilde{\MBw}^{\tp}\MBV\MBC^{\tp}\BOmega_{21}\MBQ_{11}^{-1}\MBD_{\Bbeta_1}\\
=&\MBD_{\Bbeta_1}^{\tp}\MBQ_{11}^{-1}\MBD_{\Bbeta_1}+\tilde{\MBw}^{\tp}\MBV\MBB\BOmega\MBB^{\tp}\MBV\tilde{\MBw}+2\tilde{\MBw}^{\tp}\MBV\MBB\BOmega\MBH\MBQ_{11}^{-1}\MBD_{\Bbeta_1},
\end{align*}
where $\tilde{\MBw}=(\tilde{w}_1,...,\tilde{w}_{k_2})^{\tp}$, $\MBV=\DIAG(\MBD_{\Bbeta_1}^{\tp}\BXi\MBC)=\DIAG(\MBC^{\tp}\BXi^{\tp}\MBD_{\Bbeta_1})$, $\MBB=\left[-\MBC^{\tp}\BXi^{\tp}\quad\MBC^{\tp}\right]$ and $\MBH=\left[\MBI\quad\MBzero\right]^{\tp}$. 

Taking the mean and the variance of $R_{\mu}$ together, we can have the asymptotic MSE of $R_{\mu}$ by calculating the sum of the squared mean and the variance:
\begin{align*}
\mathrm{AMSE}(\mu(\hat{\Bbeta}_{1\mathrm{WALS}}))=&\MBD_{\Bbeta_1}^{\tp}\BXi\Bdelta\Bdelta^{\tp}\BXi^{\tp}\MBD_{\Bbeta_1}+\tilde{\MBw}^{\tp}\MBV\MBC^{-1}\Bdelta\Bdelta^{\tp}\MBC\MBV\tilde{\MBw}-2\tilde{\MBw}^{\tp}\MBV\MBC^{-1}\Bdelta\Bdelta^{\tp}\BXi^{\tp}\MBD_{\Bbeta_1}\\
&+\MBD_{\Bbeta_1}^{\tp}\MBQ_{11}^{-1}\MBD_{\Bbeta_1}+\tilde{\MBw}^{\tp}\MBV\MBB\BOmega\MBB^{\tp}\MBV\tilde{\MBw}+2\tilde{\MBw}^{\tp}\MBV\MBB\BOmega\MBH\MBQ_{11}^{-1}\MBD_{\Bbeta_1}.
\end{align*}

\end{proof}

\section*{Appendix B}
In this appendix, we detail the density functions $\pi(\eta)$ and the computational procedures for the implied shrinkage weights $\omega(t)$ corresponding to the four priors discussed in the main text. Following the standard WALS framework, the Laplace prior is given by $\pi(\eta) = \frac{c}{2}\exp(-c|\eta|)$ with $c=\ln(2)$. Because it yields a closed-form solution, its implied shrinkage weight is analytically expressed as $\omega(t) = 1 - \frac{c}{t} h(t)$, where $h(t) = \frac{\exp(-ct)\Phi(t-c)-\exp(ct)\Phi(-t-c)}{\exp(-ct)\Phi(t-c)+\exp(ct)\Phi(-t-c)}$, and $\Phi(\cdot)$ is the standard normal cumulative distribution function. 

For the heavy-tailed priors, the density functions and parameter choices are specified as follows. The Cauchy prior requires no additional parameters, with density $\pi(\eta) = \frac{1}{\pi(1+\eta^2)}$. For the Pareto prior, we adopt the preferred values from the literature, setting the shape parameter $a=0.0862$ and scale parameter $c=0.0676$ for its density $\pi(\eta) = \frac{c(1-a)}{2a}(1+c|\eta|)^{-1/a}$. Similarly, to ensure sufficient tail heaviness and smoothness for the Weibull prior, we set $b=0.8876$ and $c=\ln(2)$ in its density $\pi(\eta) = \frac{bc}{2} |\eta|^{b-1} \exp(-c|\eta|^b)$. 

Unlike the Laplace prior, combining these heavy-tailed priors with the normal likelihood $\phi(t-\eta)$ does not yield a closed-form posterior mean. Therefore, we evaluate the posterior mean, $m(t) = \int_{-\infty}^{\infty} \eta \phi(t-\eta)\pi(\eta) d\eta \big/ \int_{-\infty}^{\infty} \phi(t-\eta)\pi(\eta) d\eta$, using Gaussian quadrature. We truncate the integration interval to $[-20, 20]$, which captures the vast majority of the probability mass and ensures computational stability. The implied weight is then computed as $\omega(t) = m(t)/t$. Finally, to prevent numerical division-by-zero errors near the origin ($t \to 0$), our implementation imposes a strictly positive lower bound (e.g., $10^{-10}$) on $t$.

\section*{Appendix C}

\begin{figure}[!htpb]
    \centering
    \includegraphics[width=0.4\linewidth]{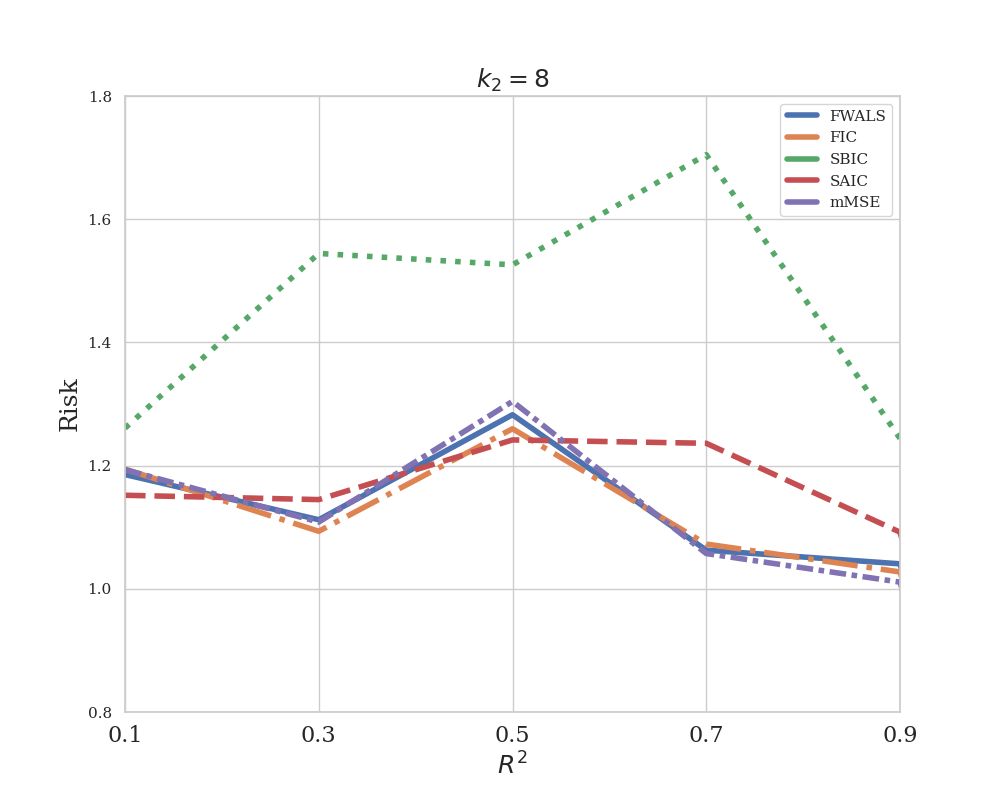}
    \includegraphics[width=0.4\linewidth]{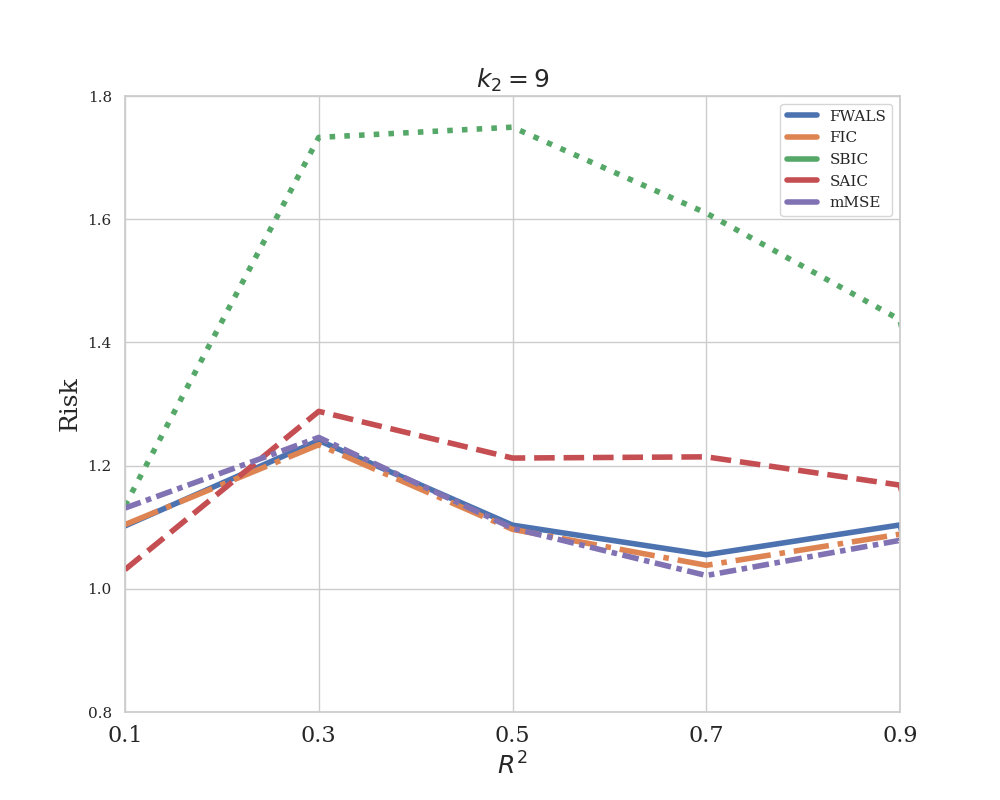}
    \includegraphics[width=0.4\linewidth]{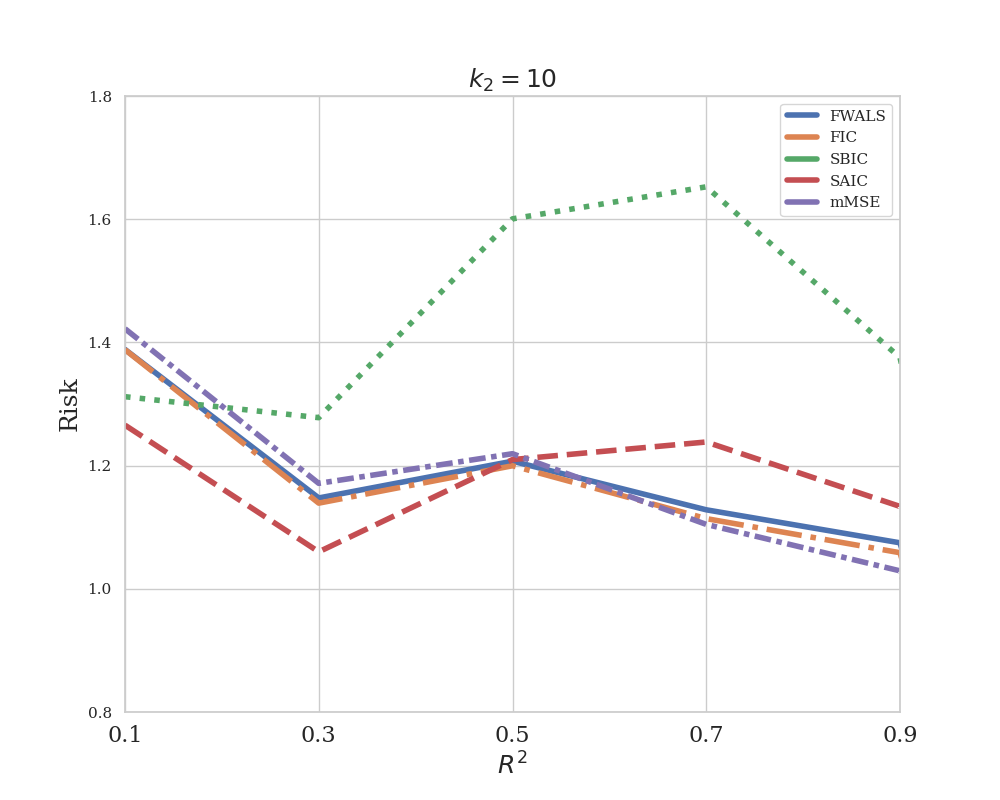}
    \includegraphics[width=0.4\linewidth]{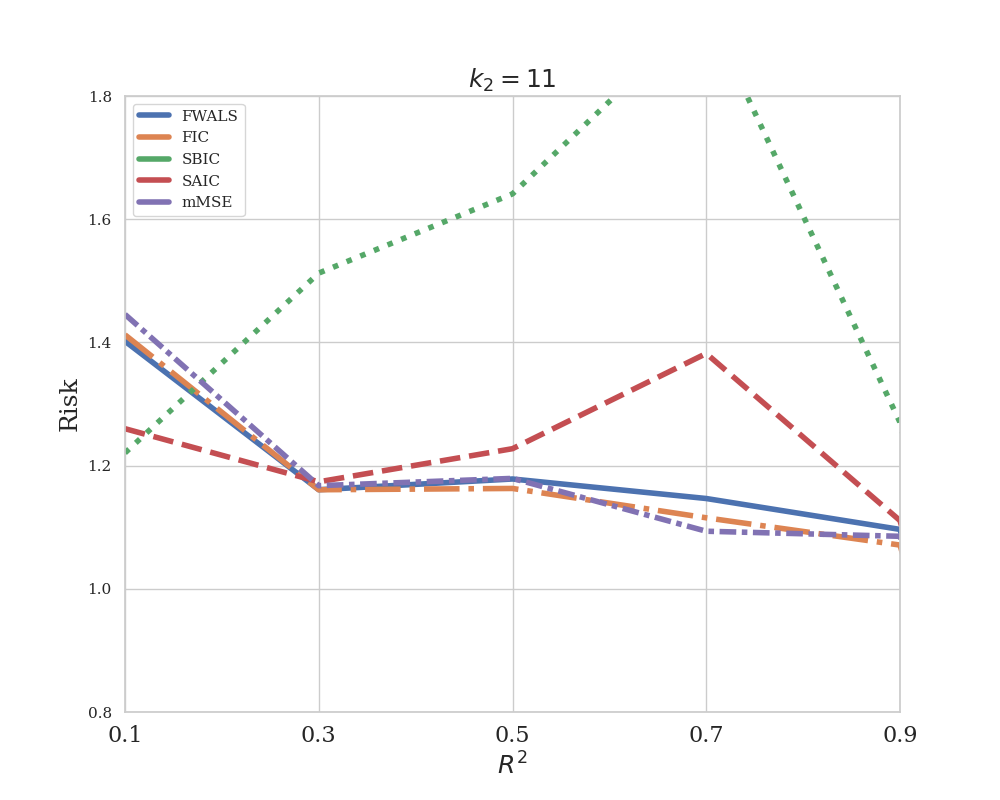}
    \caption{Risk with Different $k_2$s~($N=100$,$\tau=0.3$)}
    \label{fig10}
\end{figure}

\begin{figure}[!htpb]
    \centering
    \includegraphics[width=0.4\linewidth]{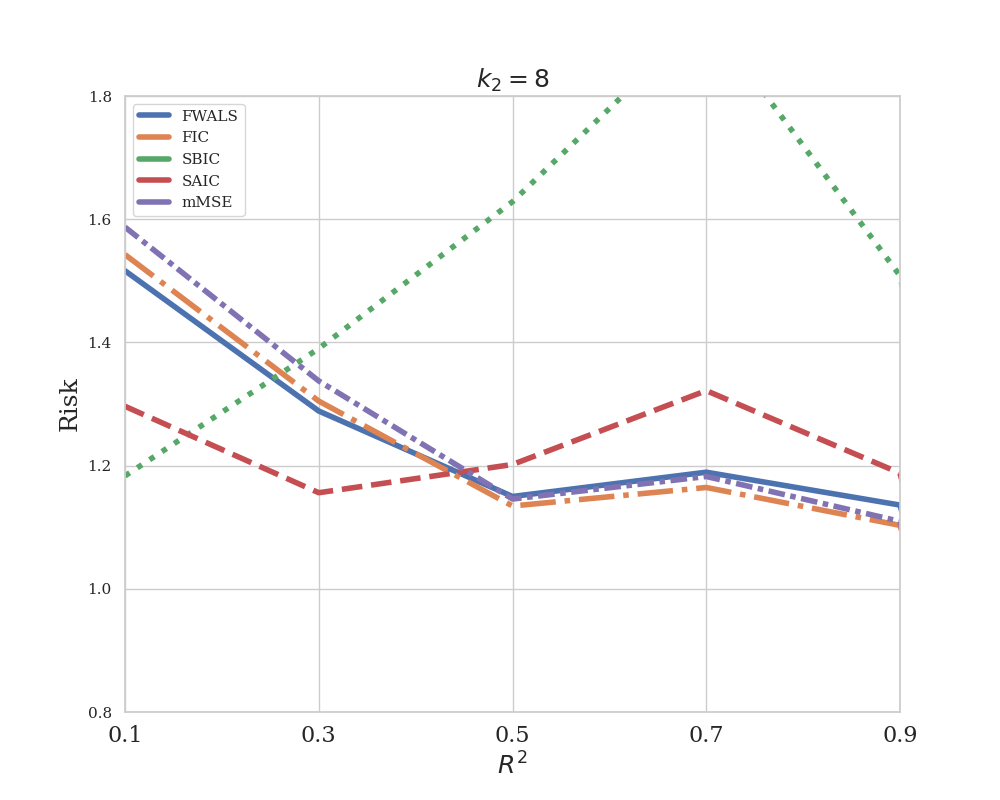}
    \includegraphics[width=0.4\linewidth]{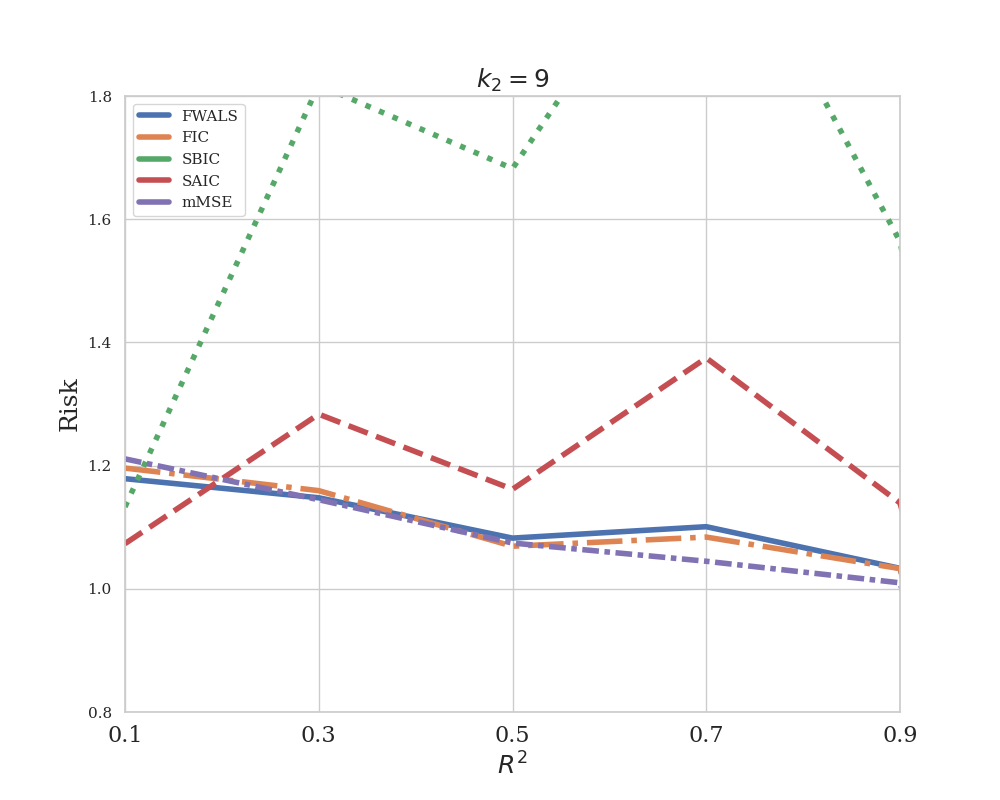}
    \includegraphics[width=0.4\linewidth]{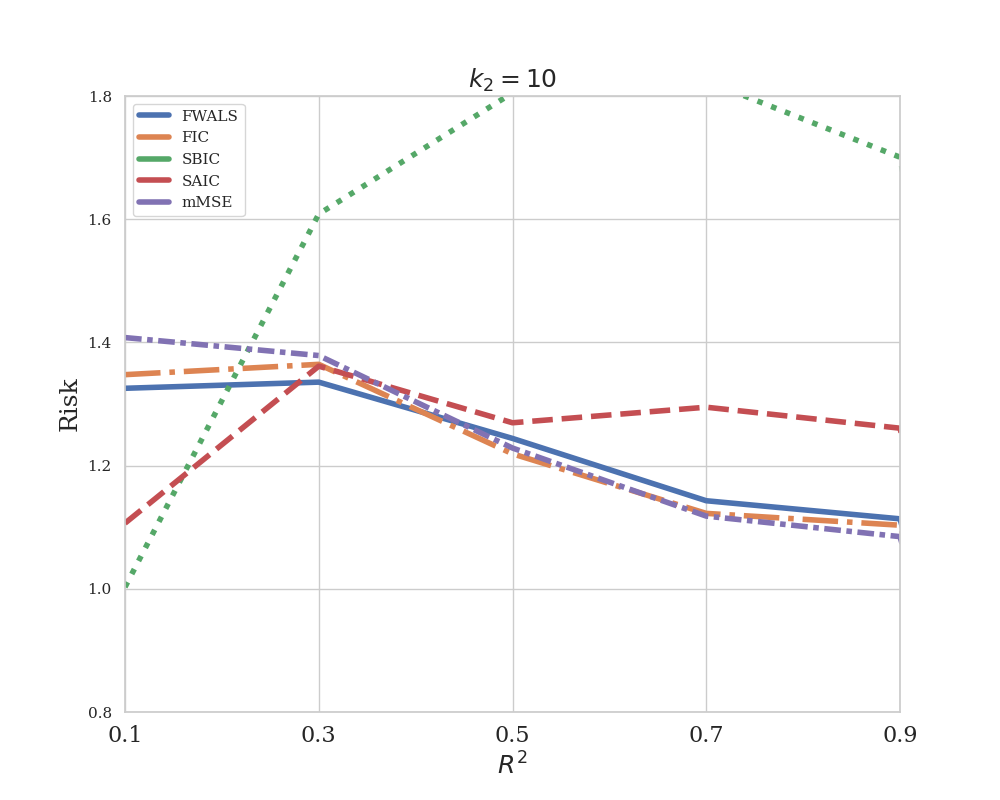}
    \includegraphics[width=0.4\linewidth]{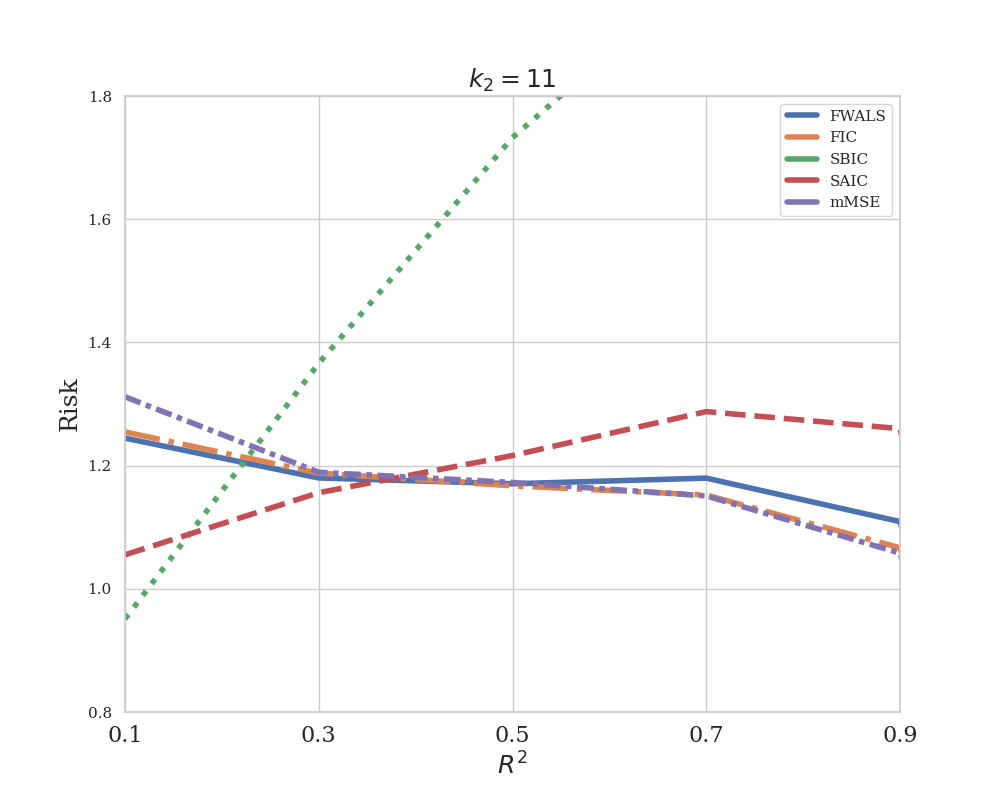}
    \caption{Risk with Different $k_2$s~($N=100$,$\tau=0.5$)}
    \label{fig11}
\end{figure}

\begin{figure}[!htpb]
    \centering
    \includegraphics[width=0.4\linewidth]{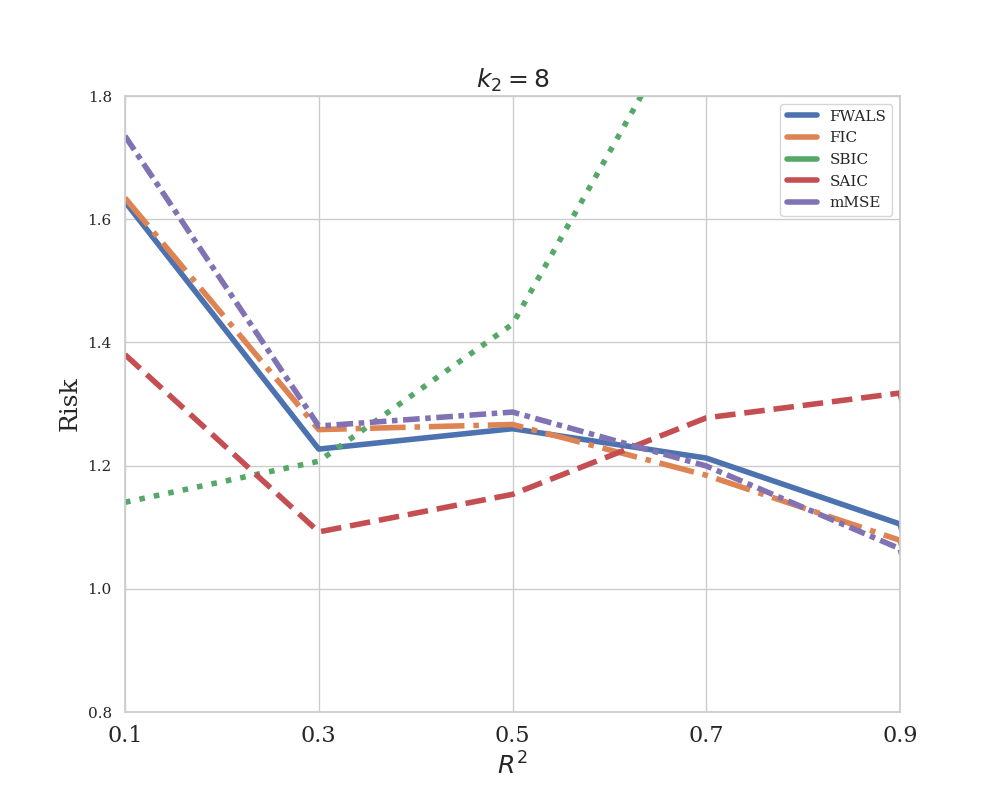}
    \includegraphics[width=0.4\linewidth]{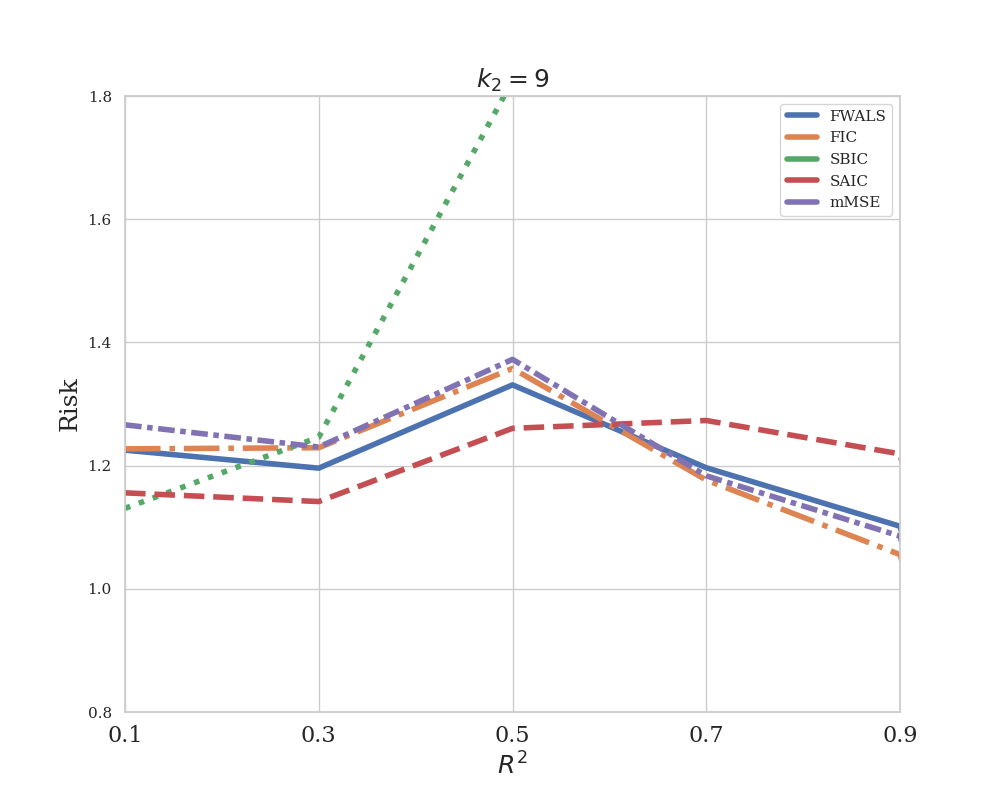}
    \includegraphics[width=0.4\linewidth]{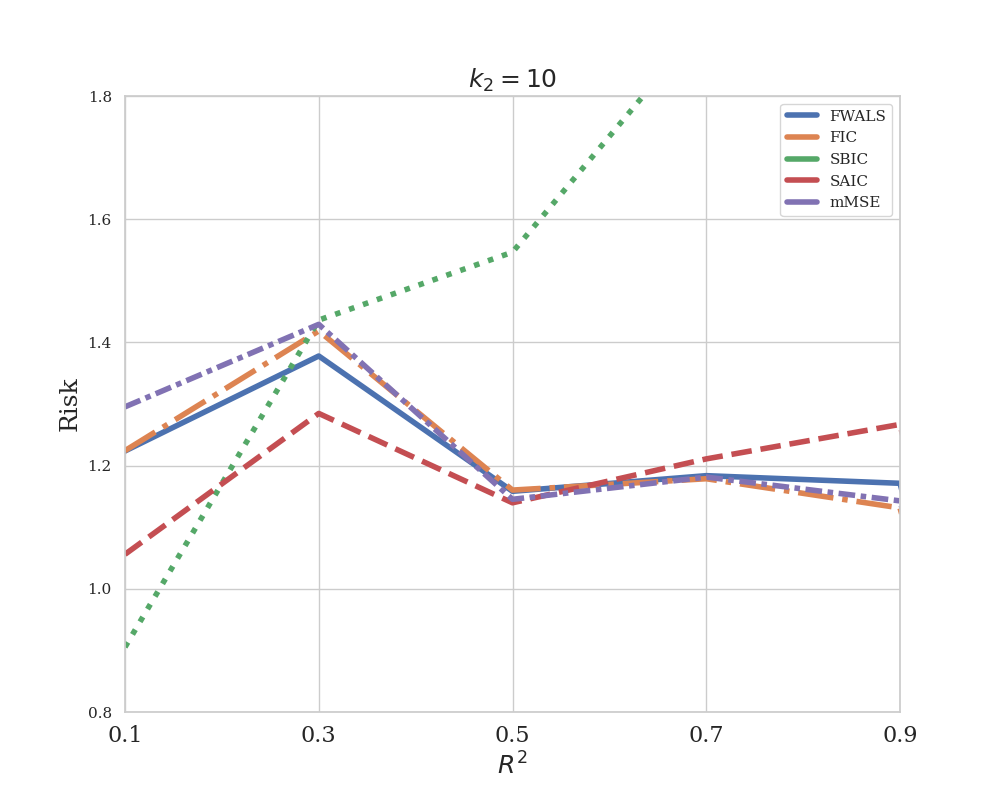}
    \includegraphics[width=0.4\linewidth]{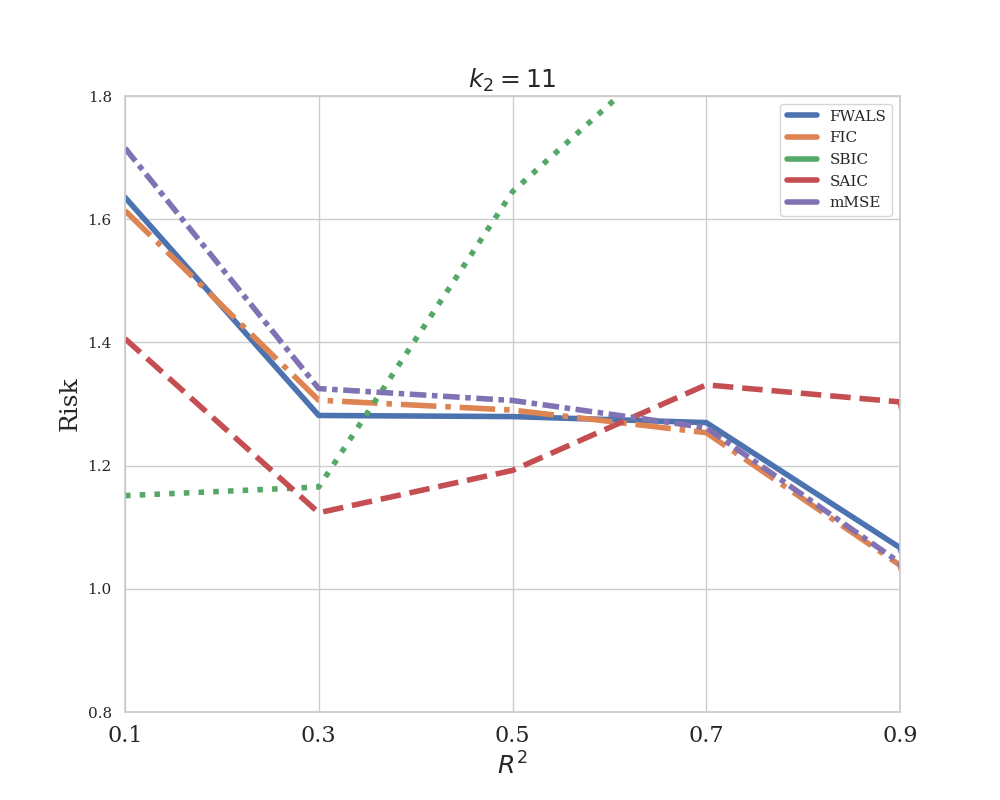}
    \caption{Risk with Different $k_2$s~($N=100$,$\tau=0.7$)}
    \label{fig12}
\end{figure}

\end{document}